\documentclass{eptcs}
\usepackage[utf8]{inputenc}
\usepackage{mdframed,float}
\usepackage{caption}
\usepackage{microtype,lmodern,amsmath,amsthm,amsfonts,longtable,etoolbox,tikz,colortbl}
\usetikzlibrary{graphs,quotes,decorations.pathreplacing,petri}
\usepackage{array}
\preto\tabular{\setcounter{magicrownumbers}{0}}
\newcounter{magicrownumbers}
\newcommand\rownumber{\stepcounter{magicrownumbers}\arabic{magicrownumbers}}
\def\rownumber{}

\newtheorem{observation}{Observation}
\newtheorem{lemma}{Lemma}
\newtheorem{theorem}{Theorem}

\newcommand{\edgeScale}{0.7}
\newcommand{\nodeScale}{0.8}
\newcommand{\myEdge}[2]{ \tikz[baseline=-1pt]{
\draw[#2,line width=0.3pt] (0,0) -- ++(0.6,0) node[anchor=base, yshift=3pt, pos=0.5] {\scalebox{\edgeScale}{$#1$}};
}}

\newcommand{\escale}[1]{\ensuremath{\textbf{\scalebox{0.7}{#1}}}}
\newcommand{\nscale}[1]{\ensuremath{\textbf{\scalebox{0.6}{#1}}}}
\newcommand{\edge}[1]{\myEdge{#1}{->}}

\newcommand{\fbedge}[1]{\myEdge{#1}{<->}}

\newcommand{\nop}{\ensuremath{\textsf{nop}}}
\newcommand{\inp}{\ensuremath{\textsf{inp}}}
\newcommand{\out}{\ensuremath{\textsf{out}}}
\newcommand{\set}{\ensuremath{\textsf{set}}}
\newcommand{\res}{\ensuremath{\textsf{res}}}
\newcommand{\swap}{\ensuremath{\textsf{swap}}}
\newcommand{\free}{\ensuremath{\textsf{free}}}
\newcommand{\used}{\ensuremath{\textsf{used}}}

\title{Tracking Down the Bad Guys: \emph{Reset} and \emph{Set} Make Feasibility for Flip-Flop Net Derivatives NP-complete}
\author{Ronny Tredup
\institute{Universit\"at Rostock, Institut f\"ur Informatik, Theoretische Informatik,\\ Albert-Einstein-Stra\ss e 22, 18059, Rostock}\email{ronny.tredup@uni-rostock.de}}
\def\titlerunning{\emph{Reset} and \emph{Set} Make Feasibility for Flip-Flop Net Derivatives NP-complete}
\def\authorrunning{Ronny Tredup}
\begin{document}
\maketitle

\begin{abstract}
Boolean Petri nets are differentiated by types of nets $\tau$ based on which of the interactions \nop, \inp, \out, \set, \res, \swap, \used, and \free\ they apply or spare.
The synthesis problem relative to a specific type of nets $\tau$ is to find a boolean $\tau$-net $N$ whose reachability graph is isomorphic to a given transition system $A$.
The corresponding decision version of this search problem is called \emph{feasibility}.
Feasibility is known to be polynomial for all types of \emph{flip flop} derivates defined by $\{\nop, \swap\}\cup \omega$, $\omega \subseteq \{\inp, \out, \used, \free\}$.
In this paper, we replace $\inp, \out$ by $\res,\set$ and show that feasibility becomes NP-complete for $\{\nop, \swap\}\cup\omega $ if $\omega \subseteq \{\res, \set, \used, \free\}$ such that $\omega \cap \{\set, \res\}\not=\emptyset$ and $\omega \cap \{\used, \free\}\not=\emptyset$.
The reduction guarantees a low degree for $A$'s states and, thus, preserves hardness of feasibility even for considerable input restrictions.
\end{abstract}

%%%%%%%%%%%%%%%%%%%%%%%%%%%%%%%%%%%%%%%%%%%%%%%%%%%%%%%%%%%%%%%%%%%%%%%%%%%%%%%

\section{Introduction}\label{sec:introduction}
Boolean Petri nets have been widely regarded as a fundamental model for concurrent systems.
These Petri nets allow at most one token per place in every reachable marking.
Accordingly, a place $p$ can be regarded as a boolean condition which is \emph{true} if $p$ contains a token and is \emph{false} if $p$ is empty, respectively.
A place $p$ and a transition $t$ of a boolean Petri net are set in relation by one of the following (boolean) \emph{interactions}: \emph{no operation} (\nop), \emph{input} (\inp), \emph{output} (\out), \set, \emph{reset} (\res), \emph{inverting} (\swap), \emph{test  if true} (\used), and \emph{test if false} (\free).
An interaction defines which pre-condition $p$ has to satisfy to activate $t$ and it determines $p$'s post-condition after $t$ has fired:
\inp\ (\out) mean that $p$ has to be \emph{true} (\emph{false}) to allow $t$'s firing and if $t$ fires then $p$ become \emph{false} (\emph{true}).
The interaction \free\ (\used) says that if $t$ is activated then $p$ is \emph{false} (\emph{true}) and $t$'s firing has no impact on $p$.
The other interactions \nop, \set, \res , \swap\ are pre-condition free, that is, neither \emph{true} nor \emph{false} prevent $t$'s firing.
Moreover, \nop\ means that the firing of $t$ has no impact and leaves $p$'s boolean value unchanged.
By \res\ (\set), $t$'s firing determine $p$ to be \emph{false} (\emph{true}).
Finally, \swap\ says that if $t$ fires then it inverts $p$'s boolean value.

Boolean Petri nets are differentiated by types of nets $\tau$ accordingly to the boolean interactions they allow.
Since we have eight interactions to choose from, this results in a total of 256 different types.
So far research has explicitly defined seven of them:
\emph{Elementary net systems} $\{\nop, \inp, \out\}$~\cite{DBLP:conf/ac/RozenbergE96},  \emph{Contextual nets} $\{\nop, \inp, \out, \used, \free\}$~\cite{DBLP:journals/acta/MontanariR95}, \emph{event/condition nets} $\{\nop, \inp, \out, \used\}$~\cite{DBLP:series/txtcs/BadouelBD15}, \emph{inhibitor nets} $\{\nop, \inp, \out, \free\}$~\cite{DBLP:conf/apn/Pietkiewicz-Koutny97}, \emph{set nets} $\{\nop, \inp, \set, \used\}$~\cite{DBLP:journals/acta/KleijnKPR13}, \emph{trace nets} $\{\nop, \inp, \out, \set, \res, \used, \free\}$~\cite{DBLP:journals/acta/BadouelD95}, and \emph{flip flop nets} $\{\nop, \inp, \out, \swap\}$~\cite{DBLP:conf/stacs/Schmitt96}.

Boolean net synthesis relative to a specific type of nets $\tau$ is the challenge to find for a given transition system $A$ (TSs, for short) a boolean $\tau$-net $N$ whose reachability graph is isomorphic to $A$.
The corresponding decision version, called $\tau$-\emph{feasibility}, asks if for a given TS $A$ a searched $\tau$-net exists.
This paper investigates the computational complexity of feasibility depending on the target type of nets $\tau$.
The complexity of $\tau$-feasibility has originally been investigated for elementary net systems~\cite{DBLP:journals/tcs/BadouelBD97}, where it is NP-complete to decide if general TSs can be synthesized.
In~\cite{DBLP:conf/concur/TredupR18,DBLP:conf/apn/TredupRW18} this has been confirmed even for considerably restricted input TSs.
Further boolean net types have been investigated in~\cite{DBLP:conf/tamc/TredupR19,DBLP:conf/apn/TredupR19}, cf. Table~\ref{tab:summary}. 

\begin{figure}[t!]\centering
\begin{tabular}{@{\makebox[1.5em][r]{\rownumber\space}}| l | c |c}
Type of net $\tau$ & Complexity status & \#
\gdef\rownumber{\stepcounter{magicrownumbers}\arabic{magicrownumbers}} \\ \hline
%polytime tamc
$\tau = \{\nop, \res\} \cup \omega$, $\omega \subseteq \{\inp, \used, \free\}$ & polynomial time & 8 \\
 $\tau = \{\nop, \set\} \cup \omega$, $\omega \subseteq \{\out, \used, \free\}$ & polynomial time  & 8 \\
$\tau = \{\nop, \swap\} \cup \omega$, $\omega \subseteq \{\inp, \out, \used, \free\}$ & polynomial time  & 16 \\
$\tau = \{\nop\} \cup \omega$, $\omega \subseteq \{\used, \free\}$ & polynomial time  & 4 \\ 
%hardness tamc
$\tau = \{\nop, \inp, \free\}$ or $\tau = \{\nop, \inp, \used, \free\}$ & NP-complete & 2 \\ 
 $\tau = \{\nop, \out, \used\}$ or  $\tau = \{\nop, \out, \used, \free\}$ & NP-complete  & 2 \\ 
$\tau = \{\nop, \set, \res \} \cup \omega$,  $\emptyset \not=\omega \subseteq \{\used, \free\}$& NP-complete  & 3 \\ \hline
%hardness ataed
$\tau = \{\nop, \inp, \out \} \cup \omega$,  $\omega \subseteq \{\used, \free\}$& NP-complete  & 4  \\ 
 $\tau = \{\nop, \inp, \res,\swap \} \cup \omega$,  $\omega \subseteq \{\used, \free\}$& NP-complete  &  4 \\ 
 $\tau = \{\nop, \out, \set,\swap \} \cup \omega$,  $\omega \subseteq \{\used, \free\}$& NP-complete  &  4 \\ 
$\tau = \{\nop, \inp, \set \} \cup \omega$,  $\omega \subseteq \{\out, \res,\swap,\used, \free\}$& NP-complete  & 24+8 \\ 
$\tau = \{\nop, \out, \res \} \cup \omega$,  $\omega \subseteq \{\inp, \set,\swap,\used, \free\}$& NP-complete  & 24+8 \\ \hline
%this paper
$\tau = \{\nop,\set,\swap,\free\}$, $\tau = \{\nop,\res,\swap,\used\}$ & NP-complete  & 2  \\ 
$\tau = \{\nop,\set,\swap,\used\}$, $\tau = \{\nop,\res,\swap,\free\}$ & NP-complete  & 2  \\ 
$\tau = \{\nop,\set,\swap,\free, \used\}$, $\tau = \{\nop,\res,\swap,\free, \used\}$ & NP-complete  & 2  \\ 
$\tau = \{\nop,\set,\res, \swap\}\cup \omega$, $\emptyset\not=\omega\subseteq \{\free,\used\}$ & NP-complete  & 3  \\ 
\end{tabular}
\caption{Summary of the complexity results for boolean net synthesis.
(1-7): Results of \cite{DBLP:conf/tamc/TredupR19} reestablishing the result for flip flop nets \cite{DBLP:conf/stacs/Schmitt96}. 
(8-12): Results of \cite{DBLP:conf/apn/TredupR19} reestablishing the result for elementary net systems~\cite{DBLP:journals/tcs/BadouelBD97}.
The rows 11 and 12 intersect in eight supersets of $\{\nop,\inp,\out,\set, \res\}$.
(13-16): Results of this paper.
Notice that isomorphic types occur in the same row.
Altogether, this paper discovers 9 \emph{new} types with an NP-hard synthesis problem.
}\label{tab:summary}
\end{figure}
In particular, feasibility is NP-complete for the seven types of nets $\{\nop, \inp, \free\}$ with optional \used\ and $\{\nop, \out, \used\}$ with optional \free\ and $\{\nop, \set, \res\}$ extended with at least one of \used\ and \free\ \cite{DBLP:conf/tamc/TredupR19}, cf. Table~\ref{tab:summary}~(5-7).

On the contrary, \cite{DBLP:conf/stacs/Schmitt96} shows that feasibility for flip flop nets, which simply extend elementary net systems by the \swap\ interaction, is decidable in polynomial time.
Moreover, $\swap$ preserves tractability of feasibility for all boolean types which includes \nop\ but excludes $\res$ and $\set$, cf. Table~\ref{tab:summary}~(3).
In view of these results, we are interested in the interactions that cause the difference between tractability and intractability of feasibility for boolean Petri nets.
In this paper, we investigate the situation where $\{\nop, \swap\}\cup \omega$ is extended by a subset $\omega\subseteq\{\res, \set, \used, \free \}$ such that: 
\begin{enumerate}
\item
The resulting type has not been investigated in \cite{DBLP:conf/tamc/TredupR19}, cf. Table~\ref{tab:summary}~(3), that is, $\omega\cap \{\res, \set\}\not=\emptyset$.

\item
The resulting type contains at least one interaction which does not allow unconditioned firing, that is, $\omega\cap \{\used, \free\}\not=\emptyset$.
\end{enumerate}
Note that $\{\res, \set, \used, \free \}$ and $\{\inp, \out, \used, \free \}$ only differ in the replacement of $\inp, \out$ by $\res, \set$.
However, in this paper, we show that feasibility becomes difficult for the 9 extensions of $\{\nop, \swap\}$ by at least one interaction from both, $\{\set, \res\}$ and $\{\used, \free\}$.
For one thing, as feasibility for $\{\nop, \swap\}\cup \omega$, $\omega\subseteq\{\inp, \out, \used, \free \}$, is doable in polynomial time \cite{DBLP:conf/stacs/Schmitt96,DBLP:conf/tamc/TredupR19}, our result exhibits that $\res$ and $\set$ are the \emph{bad guys} which make the decision problem computationally complex.
For another thing, it shows that in this context the interactions $\res$ and $\set$ are strictly more powerful than $\inp$ and $\out$ as their takeover allow to encode NP-complete problems.

Our result is robust with respect to considerable input restrictions.
In particular, we introduce \emph{grade} as a parameter of TSs as it has been done in~\cite{DBLP:conf/concur/TredupR18,DBLP:conf/apn/TredupR19,DBLP:conf/apn/TredupRW18}.
In a $g$-grade TS the number of outgoing and incoming transitions, respectively, is limited by $g$ for every state.
This is a very natural parameter as most parts of real world TSs are often heavily restricted with respect to their grade.
According to~\cite{jordiCortadella2017}, benchmarks in digital hardware design, for instance, often show TSs with few choices, that is, with a low grade.
Nevertheless, our method demonstrates that restricting the input to TSs with small grade has little influence on the complexity of feasibility for all 9 covered net types.
In particular, we show that feasibility for net type extensions of $\{\nop, \swap\}$ by at least one from both $\{\set, \res\}$ and $\{\used, \free\}$ remains NP-complete if only $g$-grade TSs are considered for any $g \geq 2$.

To simplify our argumentation we detach our notions from Petri nets and focus on TSs.
For this purpose, we use the well known equality between feasibility and the conjunction of the so called state separation property (SSP) and the event state separation property (ESSP)~\cite{DBLP:series/txtcs/BadouelBD15}, which are solely defined on the input TSs.
The presented polynomial time reduction translates the NP-complete cubic monotone one-in-three $3$-SAT problem~\cite{DBLP:journals/dcg/MooreR01} into the ESSP of the considered 9 boolean net types.
As we also make sure that given boolean expressions $\varphi$ are transformed to TSs $A(\varphi)$ where the ESSP relative to the considered type implies the SSP, we always show the NP-completeness of the ESSP and feasibility at the same time.
Instead of 9 individual proofs, our approach covers all cases by just two reductions following a common pattern.

%%%%%%%%%%%%%%%%%%%%%%%%%%%
\section{Preliminary Notions}\label{sec:preliminaries}%
%%%%%%%%%%%%%%%%%%%%%%%%%%%

This section provides short formal definitions of all preliminary notions used in the paper.
A \emph{transition system} (TS, for short) $A=(S,E, \delta)$ is a directed labeled graph with nodes $S$, events $E$ and partial transition function $\delta: S\times E \longrightarrow S$, where $\delta(s,e)=s'$ is interpreted as $s\edge{e}s'$.
By $s\fbedge{e}s'$ we denote the fact that $s\edge{e}s'$ and $s'\edge{e}s$ are present.
An event $e$ \emph{occurs} at a state $s$, denoted by $s\edge{e}$, if $\delta(s,e)$ is defined.
An \emph{initialized} TS $A=(S,E,\delta, s_0)$ is a TS with a distinct state $s_0\in S$.
TSs in this paper are \emph{deterministic} by design as their state transition behavior is given by a (partial) function.
Initialized TSs are also required to make every state \emph{reachable} from $s_0$ by a directed path.

\begin{figure}[H]\centering
\begin{tabular}{c|c|c|c|c|c|c|c|c}
$x$ & $\nop(x)$ & $\inp(x)$ & $\out(x)$ & $\set(x)$ & $\res(x)$ & $\swap(x)$ & $\used(x)$ & $\free(x)$\\ \hline
$0$ & $0$ & & $1$ & $1$ & $0$ & $1$ & & $0$\\
$1$ & $1$ & $0$ & & $1$ & $0$ & $0$ & $1$ & \\
\end{tabular}
\caption{
All interactions in $I$.
An empty cell means that the column's function is undefined on the respective $x$.
The entirely undefined function is missing in $I$.
}\label{fig:interactions}
\end{figure}

A (boolean) \emph{type of nets} $\tau=(\{0,1\},E_\tau,\delta_\tau)$ is a TS such that $E_\tau$ is a subset of the boolean interactions:
$E_\tau \subseteq I = \{\nop, \inp, \out, \set, \res, \swap, \used, \free\}$. 
The interactions $i \in I$ are binary partial functions $i: \{0,1\} \rightarrow \{0,1\}$ as defined in the listing of Figure~\ref{fig:interactions}.
For all $x\in \{0,1\}$ and all $i\in E_\tau$ the transition function of $\tau$ is defined by $\delta_\tau(x,i)=i(x)$.
Notice that $I$ contains all possible binary partial functions $\{0,1\} \rightarrow \{0,1\}$ except for the entirely undefined function $\bot$. 
Even if a type $\tau$ includes $\bot$, this event can never occur, so it would be useless.
Thus, $I$ is complete for deterministic boolean types of nets, and that means there are a total of 256 of them.
By definition, a (boolean) type $\tau$ is completely determined by its event set $E_\tau$.
Hence, in the following we will identify $\tau$ with $E_\tau$, cf. Figure~\ref{fig:types}.

%%%%%%%%%%
\begin{figure}[h!]

\centering
\begin{minipage}{1\textwidth}
\centering
\begin{tikzpicture}[scale = 1.2]
\begin{scope}%nop, out, res, free, swap
\node (0) at (0,0) {\scalebox{\nodeScale}{$0$}};
\node (1) at (2,0) {\scalebox{\nodeScale}{$1$}};

\path (0) edge [->, out=-120,in=120,looseness=5] node[left, align =left] {\scalebox{\edgeScale}{$\nop$} \\ \scalebox{\edgeScale}{\free}} (0);
\path (1) edge [<-, out=60,in=-60,looseness=5] node[right, align=left] {\scalebox{\edgeScale}{$\nop$} \\ \scalebox{\edgeScale}{$\set$}  } (1);

\path (0) edge [<-, bend right= 30] node[below] {\scalebox{\edgeScale}{$\swap$}} (1);
\path (0) edge [->, bend left= 30] node[above] {\scalebox{\edgeScale}{\set, \swap}} (1);
\end{scope}
\begin{scope}[xshift=4.5cm]%nop, inp, set, used, swap
\node (0) at (0,0) {\scalebox{\nodeScale}{$0$}};
\node (1) at (2,0) {\scalebox{\nodeScale}{$1$}};

\path (0) edge [->, out=-120,in=120,looseness=5] node[left, align =left] {\scalebox{\edgeScale}{\nop} \\ \scalebox{\edgeScale}{\res} } (0);
\path (1) edge [<-, out=60,in=-60,looseness=5] node[right, align=left] {\scalebox{\edgeScale}{$\nop$} \\ \scalebox{\edgeScale}{\used}} (1);

\path (0) edge [<-, bend right= 30] node[below] {\scalebox{\edgeScale}{\res, \swap}} (1);
\path (0) edge [->, bend left= 30] node[above] {\scalebox{\edgeScale}{\swap}} (1);
\end{scope}
\end{tikzpicture}
\caption{
Left: $\tau=\{\nop, \set, \swap, \free\}$.
Right: $\tilde{\tau}=\{\nop, \res, \swap, \used\}$.
$\tau$ and $\tilde{\tau}$ are isomorphic.
The isomorphism $\phi: \tau\rightarrow \tilde{\tau}$ is given by $\phi(s)=1-s$ for $s\in \{0,1\}$, $\phi(i)=i$ for $i\in \{\nop,\swap\}$, $\phi(\res)=\set$ and $\phi(\free)=\used$.}
\label{fig:types}
\end{minipage}

\vspace{0.5cm}
%%%%%%%%%%%%%%%%
\begin{minipage}{1\textwidth}%
%%%%%%%%%%%%%%%%
%
\centering
\begin{minipage}{0.2\textwidth}
\begin{tikzpicture}[new set = import nodes]
\begin{scope}[nodes={set=import nodes}]%SSP & ESSP

\node (0) at (0,0) {\scalebox{\nodeScale}{$s_0$}};
\node (1) at (1.5,0) {\scalebox{\nodeScale}{$s_1$}};
\node (2) at (3,0) {\scalebox{\nodeScale}{$s_2$}};
\graph {
 (0)->[ "\escale{$a$}"] (1)<->["\escale{$a$}"] (2);%
};
\end{scope}
\end{tikzpicture}
\caption*{TS $A_1$.}
\end{minipage}
\hspace{0.3cm}
\begin{minipage}{0.2\textwidth}
\begin{tikzpicture}[new set = import nodes]
\begin{scope}[nodes={set=import nodes}]%SSP but not ESSP

\node (0) at (0,0) {\scalebox{\nodeScale}{$s_0$}};
\node (1) at (1.5,0) {\scalebox{\nodeScale}{$s_1$}};
\node (2) at (3,0) {\scalebox{\nodeScale}{$s_2$}};
\graph {
 (0)->[ "\escale{$a$}"] (1)->["\escale{$a$}"] (2);%
};
\end{scope}
\end{tikzpicture}
\caption*{TS $A_2$.}
\end{minipage}
\hspace{0.3cm}
\begin{minipage}{0.2\textwidth}
\begin{tikzpicture}[new set = import nodes]
\begin{scope}[nodes={set=import nodes}]%ESSP but not SSP

\node (0) at (0,0) {\scalebox{\nodeScale}{$s_0$}};
\node (1) at (1.5,0) {\scalebox{\nodeScale}{$s_1$}};
\node (2) at (3,0) {\scalebox{\nodeScale}{$s_2$}};
\graph {
 (0)<->[ "\escale{$a$}"] (1)<->["\escale{$a$}"] (2);%
};
\end{scope}
\end{tikzpicture}
\caption*{TS $A_3$.}
\end{minipage}
\hspace{0.3cm}
\begin{minipage}{0.2\textwidth}
\begin{tikzpicture}[new set = import nodes]
\begin{scope}[nodes={set=import nodes}]%not SSP and not ESSP

\node (0) at (0,0) {\scalebox{\nodeScale}{$s_0$}};
\node (1) at (1,0) {\scalebox{\nodeScale}{$s_1$}};
\node (2) at (2,0) {\scalebox{\nodeScale}{$s_2$}};
\node (3) at (3,0) {\scalebox{\nodeScale}{$s_3$}};
\graph {
 (0)->[ "\escale{$a$}"] (1)->["\escale{$a$}"] (2)->["\escale{$a$}"] (3);%
};
\end{scope}
\end{tikzpicture}
\caption*{TS $A_4$.}
\end{minipage}
\caption{Let $\tau=\{\nop,\set, \swap, \free\}$.
The TSs $A_1,\dots, A_4$ give examples for the presence and absence of the $\tau$-(E)SSP, respectively:
TS $A_1$ has the $\tau$-ESSP as $a$ occurs at every state.
It has also the $\tau$-SSP: 
The region $R=(sup, sig)$ where $sup(s_0)=sup(s_2)=1$, $sup(s_1)=0$ and $sig(a)=\swap$ separates the pairs $s_0,s_1$ and $s_2, s_1$.
Moreover, the region $R'=(sup', sig')$ where $sup'(s_0)=0$ and $sup'(s_1)=sup'(s_2)=1$ and $sig'(a)=\set$ separates $s_0$ and $s_1$. 
Notice that $R$ and $R'$ can be translated into $\tilde{\tau}$-regions, where $\tilde{\tau}=\{\nop, \res, \swap, \used\}$, via the isomorphism of Figure~\ref{fig:types}.
For example, if $s\in S(A_1)$ and $e\in E(A_1)$ and $sup''(s)=\phi(sup(s))$ and $sig''(e)=\phi(sig(e))$ then the resulting $\tilde{\tau}$-region $R''=(sup'', sig'')$ separates $s_0, s_1$ and $s_2,s_1$.
Thus, $A_1$ is also $\tilde{\tau}$-feasible. 
The other TSs are not $\tau$-feasible:
TS $A_2$ has the $\tau$-SSP but not the $\tau$-ESSP as event $a$ is not inhibitable at the state $s_2$.
TS $A_3$ has the $\tau$-ESSP ($a$ occurs at every state) but not the $\tau$-SSP as $s_1$ and $s_2$ are not separable.
TS $A_4$ has neither the $\tau$-ESSP nor the $\tau$-SSP.
}
\label{fig:regions}
\end{minipage}
\end{figure}

A boolean Petri net $N = (P, T, M_0, f)$ of type $\tau$ ($\tau$-net, for short) is given by finite and disjoint sets $P$ of places and $T$ of transitions, an initial marking $M_0: P\longrightarrow  \{0,1\}$ and a (total) flow function $f: P \times T \rightarrow \tau$. 
The meaning of a boolean net is to realize a certain behavior by firing sequences of transitions. 
In particular, a transition $t \in T$ can fire in a marking $M: P\longrightarrow  \{0,1\}$ if $\delta_\tau(M(p), f(p,t))$ is defined for all $p\in P$.
By firing, $t$ produces the next marking $M' : P\longrightarrow  \{0,1\}$ where $M'(p)=\delta_\tau(M(p), f(p,t))$ for all $p\in P$. 
This is denoted by $M \edge{t} M'$.
Given a $\tau$-net $N=(P, T, M_0, f)$, its behavior is captured by a transition system $A(N)$, called the \emph{reachability graph} of $N$.
The state set of $A(N)$ consists of all markings that, starting from initial state $M_0$, can be reached by firing a sequence of transitions.
For every reachable marking $M$ and transition $t \in T$ with $M \edge{t} M'$ the state transition function $\delta$ of $A$ is defined as $\delta(M,t) = M'$.

Boolean net synthesis for a type $\tau$ is going backwards from input TS $A=(S, E, \delta, s_0)$ to the computation of a $\tau$-net $N$ with $A(N)$ isomorphic to $A$, if such a net exists.
In contrast to $A(N)$, the abstract states $S$ of $A$ miss any information about markings they stand for.
Accordingly, the events $E$ are an abstraction of $N$'s transitions $T$ as they relate to state changes only globally without giving the information about the local changes to places.
After all, the transition function $\delta: S\times E\rightarrow S$ still tells us how states are affected by events.

In this paper, we investigate the computational complexity of the corresponding decision version: 
$\tau$-\emph{feasibility} is the problem to decide the existence of a $\tau$-net $N$ with $A(N)$ isomorphic to the given TS $A$.
To describe $\tau$-feasibility without referencing the searched $\tau$-net $N$, in the sequel, we introduce the $\tau$-\emph{state separation property} ($\tau$-\emph{SSP}, for short) and the $\tau$-\emph{event state separation property} ($\tau$-\emph{ESSP}, for short) for TSs.
In conjunction, they are equivalent to $\tau$-feasibility.
The following notion of $\tau$-regions allows us to define the announced properties, cf. Figure~\ref{fig:regions}.

A (boolean) $\tau$-region of a TS $A=(S, E, \delta, s_0)$ is a pair $(sup, sig)$ of \emph{support} $sup: S \rightarrow  \{0,1\}$ and \emph{signature} $sig: E \rightarrow \tau$ where every transition $s \edge{e} s'$ of $A$ leads to a transition $sup(s) \edge{sig(e)} sup(s')$ of $\tau$.
While a region divides $S$ into the two sets $sup^{-1}(b) = \{s \in S \mid sup(s) = b\}$ for $b \in \{0,1\}$, the events are cumulated by $sig^{-1}(i) = \{e \in E \mid sig(e) = i\}$ for all available interactions $i \in \tau$.
We also use $sig^{-1}(\tau') = \{e \in E \mid sig(e) \in \tau'\}$ for $\tau' \subseteq \tau$.

For a TS $A=(S, E, \delta, s_0)$ and a type of nets $\tau$, a pair of states $s \not= s' \in S$ is $\tau$-\emph{separable} if there is a $\tau$-region $(sup, sig)$ such that $sup(s) \not= \sup(s')$.
Accordingly, $A$ has the $\tau$-SSP if all pairs of distinct states of $A$ are $\tau$-separable. 
Secondly, an event $e \in E$ is called $\tau$-\emph{inhibitable} at a state $s \in S$ if there is a $\tau$-region $(sup, sig)$ where $sup(s) \edge{sig(e)}$ does not hold, that is, the interaction $sig(e) \in \tau$ is not defined on input $sup(s) \in \{0,1\}$.
$A$ has the $\tau$-ESSP if for all states $s \in S$ it is true that all events $e \in E$ that do not occur at $s$, meaning $\neg s \edge{e}$, are $\tau$-inhibitable at $s$.
It is well known from~\cite{DBLP:series/txtcs/BadouelBD15} that a TS $A$ is $\tau$-feasible, that is, there exists a $\tau$-net $N$ with $A(N)$ isomorphic to $A$, if and only if $A$ has $\tau$-SSP \emph{and} the $\tau$-ESSP.
Types of nets $\tau$ and $\tilde{\tau}$ have an isomorphism $\phi$ if $s \edge{i} s'$ is a transition in $\tau$ if and only if $\phi(s) \edge{\phi(i)} \phi(s')$ is one in $\tilde{\tau}$.
By the following lemma, we benefit from the eight isomorphisms that map \nop\ to \nop, \swap\ to \swap, \inp\ to \out, \set\ to \res, \used\ to \free, and vice versa, cf. Figure~\ref{fig:regions}:
\begin{lemma}[Without proof]
\label{lem:isomorphic_types}
If $\tau$ and $\tilde{\tau}$ are isomorphic types of nets then a TS $A$ has the $\tau$-(E)SSP if and only if it has the $\tilde{\tau}$-(E)SSP.
\end{lemma}

%%%%%%%%%%%%%%%%%%%%%%%%%
\section{Our Contribution}\label{sec:main_result}%
%%%%%%%%%%%%%%%%%%%%%%%%%

The following theorem states our main result and covers 9 new boolean types of nets, cf. Table~\ref{tab:summary}:
\begin{theorem}\label{the:main_result} 
Let $\tau_1 = \{\nop, \set, \swap\}$ and $\tilde{\tau}_1 = \{\nop, \res, \swap\}$.
Deciding $\tau$-feasibility for $g$-grade transition systems is NP-complete if $\tau = \tau' \cup \omega$ for $\tau' \in \{\tau_1, \tilde{\tau}_1, \tau_1\cup  \tilde{\tau}_1\}$ with non empty $\omega \subseteq \{\used, \free\}$ and $g \geq 2$. 
\end{theorem}
On the one hand, it is straight forward that $\tau$-feasibility is a member of NP for all considered type of nets $\tau$:
In a non-deterministic computation, one can simply guess and check in polynomial time for all pairs $s, s'$ of states, respectively for all required pairs $s,e$ of state and event, the region that separates $s$ and $s'$, respectively inhibits $e$ at $s$, or refuse the input if such a region does not exist.

On the other hand, it is probably impossible to show hardness in NP for all considered types $\tau$ with the same reduction.
Here, we manage to boil it down to two reductions using the NP-complete cubic monotone one-in-three-$3$-SAT problem~\cite{DBLP:journals/dcg/MooreR01}.
While we start for both reductions from a common principle, the first peculiarity (reduction) of this principle is dedicated to the type of $\sigma_1$ and the second to the types of $\sigma_2$:
\[ \sigma_1= \{ \{\nop,\set,\swap,\free\}  \}, \ \  \sigma_2 =\{\{\nop,\set,\swap,\used\} \cup \omega \mid \omega \subseteq \{\res, \free\} \} \] 
Every remaining type $\tau'$ that is also covered by Theorem~\ref{the:main_result} is a member of the following set $\overline{\sigma}$:
\[\overline{\sigma} = \{ \{\nop,\res,\swap\}\cup\omega\mid  \emptyset\not=\omega\subseteq \{\free,\used\}\}\cup \{ \{\nop,\set, \res,\swap,\free\}  \}\]
For every $\tau'\in\overline{\sigma}$ there is a type $\tau\in \sigma_1\cup \sigma_2$ such that $\tau\cong \tau'$, cf. Table~\ref{tab:summary}.
Thus, by Lemma~\ref{lem:isomorphic_types}, the NP-completeness of $\tau$-feasibility for all $\tau\in \sigma_1\cup\sigma_2$ implies the NP-completeness of $\tau'$-feasibility for $\tau'\in \overline{\sigma}$ and, consequently, proves our main result.

In accordance to the just introduced approach, in the remainder of this paper we deal only with the types $\tau$ covered by $\sigma_1$ and $\sigma_2$.
This allows us to present $\tau$-regions in a compressed way, which is subject of the following section.

%%%%%%%%%%%%%%%%%%%%%%%%%%%%%%%%%
\subsection{Presenting Regions}\label{sec:compressed}%
%%%%%%%%%%%%%%%%%%%%%%%%%%%%%%%%

In this section, we introduce for $\tau\in\sigma_1\cup\sigma_2$ a concept to present a $\tau$-region $(sup, sig)$ of a TS $A$ simply by its support $sup$.
This concept is tailored to the types and TSs that concern our reductions.
More exactly, by a little abuse of notation, we identify the support $sup$ with the set of states $sup\subseteq S(A)$ which it maps to one: $sup=\{s\in S(A)\mid sup(s)=1\}$.
Moreover, every presented support $sup$ allows a signature $sig$ that gets along with the interactions $\nop,\set,\swap,\free$ if $\tau\in \sigma_1$ and $\nop,\set,\swap,\used$ if $\tau\in \sigma_2$ (the interactions which are shared by all types of $\sigma_2$).
In particular, the support allows the signature $sig$ which is for all $e\in E(A)$ defined as follows:
\[sig(e)=
\begin{cases}
\free	& \text{if } \tau\in\sigma_1 \text{ and }  sup(s)=sup(s')=0 \text{ for all } s\edge{e}s'\in A \\
\used	& \text{if } \tau\in\sigma_2 \text{ and }  sup(s)=sup(s')=1 \text{ for all } s\edge{e}s'\in A \\
\set, 	& \text{if }  s\edge{e}s'\fbedge{e}s''\in A \text{ such that } sup(s)=0 \text{ and } sup(s')=sup(s'')=1 \ (*)\\
\swap, & \text{if } s\edge{e}s'\in A \text{ and } sup(s)\not=sup(s') \text{ and $(*)$ is not true} \\
\nop, 	& \text{otherwise} 
\end{cases}
\]
We emphasize again that in accordance to this concept the signature only depends on the given support $sup$ and the set $\sigma\in \{ \sigma_1,\sigma_2\}$.
Therefore, for the sake of simplicity, in the sequel we often refer to a given support $sup$ as to the region $(sup,sig)$ which it allows and, e.g., say $sup$ inhibits $e$ at $s$ instead of $(sup,sig)$ inhibits $e$ at $s$.
\newtheorem{example}{Example}
\begin{example}
The region $R$ of TS $A_1$ that is given in Figure~\ref{fig:regions} is defined by $sup=\{s_0,s_2\}$.
Similarly, the region $R'$ is defined by $sup'=\{s_1,s_2\}$.
Finally, the support $sup'''=\emptyset$ defines a region of $R'''$ of $A_1$ where $sig'''(a)=\free$.
\end{example}
Before we can set out the details of our reductions, the following subsection introduces our way of easily generating and combining gadget TSs for our NP-completeness proofs.

%%%%%%%%%%%%%%%%%%%%%%%%%%%%%%%
\subsection{Unions of Transition Systems}\label{sec:unions}%
%%%%%%%%%%%%%%%%%%%%%%%%%%%%%%%

If $A_0=(S_0,E_0,\delta_0,s_0^0), \dots ,A_n=(S_n,E_n,\delta_n,s_0^n)$ are TSs with pairwise disjoint states (but not necessarily disjoint events) we say that $U(A_0, \dots, A_n)$ is their \emph{union}.
If $U$ contains only $g$-grade TSs for some $g\in\mathbb{N}$ then we say $U$ is $g$-grade.
By $S(U)$, we denote the entirety of all states in $A_0, \dots, A_n$ and $E(U)$ is the aggregation of all events.
For a flexible formalism, we allow to build unions recursively:
Firstly, we allow empty unions and identify every TS $A$ with the union containing only $A$, that is, $A = U(A)$.   
Next, if $U_1= U(A^1_0,\dots,A^1_{n_1}), \dots, U_m=(A^m_0,\dots,A^n_{n_m})$ are unions (possibly with $U_i =U()$ or $U_i=A_i$) then $U(U_1, \dots, U_m)$ is the union $U(A^1_0, \dots, A^1_{n_1},\dots, A^m_0, \dots, A^n_{n_m})$.

We lift the concepts of regions, SSP and ESSP to unions $U = U(A_0, \dots, A_n)$ as follows:
A (boolean) $\tau$-region $(sup, sig)$ of $U$ consists of $sup: S(U) \rightarrow  \{0,1\}$ and $sig: E(U) \rightarrow \tau$ such that, for all $i \in \{0, \dots, n\}$, the projections $sup_i(s) = sup(s), s \in S_i$ and $sig_i(e) = sig(e), e \in E_i$ provide a region $(sup_i, sig_i)$ of $A_i$.
Then, $U$ has the $\tau$-SSP if for all different states $s, s' \in S(U)$ of the \emph{same} TS $A_i$ there is a $\tau$-region $(sup,sig)$ of $U$ with $sup(s) \not= sup(s')$.
Moreover, $U$ has the $\tau$-ESSP if for all events $e \in E(U)$ and \emph{all} states $s \in S(U)$ where $s\edge{e}$ does not hold there is a $\tau$-region $(sup,sig)$ of $U$ where $sup(s) \edge{sig(e)}$ does not hold.
Naturally, $U$ is called $\tau$-feasible if it has both the $\tau$-SSP and the $\tau$-ESSP.

To merge a union $U = U(A_0, \dots, A_n)$ into a single TS, we define the the so called \emph{joining} $A(U)$:
If $s^0_0, \dots, s^n_0$ are the initial states of $U$'s TSs then $A(U) = (S(U) \cup \bot \cup \top, E(U) \cup \odot \cup \otimes \cup \ominus \cup \oplus, \delta, \bot_0)$ is a TS with fresh events $\odot=\{\odot_0, \dots, \odot_n\}$,  $\otimes=\{\otimes_0, \dots, \otimes_n\}$, $\ominus=\{\ominus_0, \dots, \ominus_n\}$, $\oplus=\{\oplus_0, \dots, \oplus_n\}$ and additional connector states $\top=\{\top_{0,1}, \top_{0,2},\top_{0,3},\dots, \top_{n,1}, \top_{n,2},\top_{n,3}\}$ and $\bot=\{\bot_0, \dots, \bot_{4(n+1)}\}$.
The TS $A(U)$ joins the individual TSs of $U$ by the partial function $\delta$, which is defined as follows:
If $s\in S(A_i)$ and $e\in E(A_i)$ then $\delta(s,e)=\delta_i(s,e)$.
Moreover, we define $\delta(\top_{i,3}, \oplus_i)=s^i_0$ and $\delta(s^i_0, \oplus_i)= \top_{i,3}$ for all $i\in \{0,\dots, n\}$.
Finally, for all $i\in \{0,\dots, n\}$ the definition of $\delta$ on the states $\bot_{4i}, \dots, \bot_{4(i+1)}, \top_{i,1},\dots, \top_{i,3}$ and events $\otimes_i,\odot_i,\ominus_i,\oplus_i$ is shown in Figure~\ref{fig:joining}.
The function $\delta$ remains undefined on all other pairs $(s,e)\in S(A(U)) \times E(A(U))$.
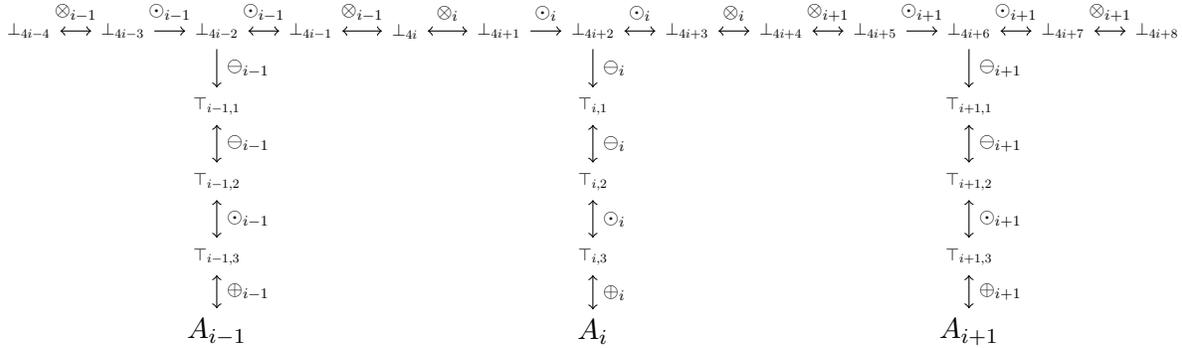
\begin{figure}[t!]
\centering
\begin{tikzpicture}[new set = import nodes]
\begin{scope}[nodes={set=import nodes}]
\foreach \i in {0,...,12} {\coordinate (\i) at (\i*1.25cm,0);}
%\foreach \i in {2,3,4}{\fill[red!15] (\i) circle (0.35cm);}
\foreach \i in {4,...,1} { \pgfmathparse{ int(4-\i)  }\node (t\pgfmathresult) at (\pgfmathresult) { \nscale{$\bot_{4i-\i}$}}; }
\foreach \i in {4} {  \node (t\i) at (\i) { \nscale{$\bot_{4i}$}}; }
\foreach \i in {1,...,8} { \pgfmathparse{ int(4+\i)  }\node (t\pgfmathresult) at (\pgfmathresult) {    \nscale{$\bot_{4i+\i}$}};    }
%%%%%%%%%%%%
\foreach \i in {1,...,4} {   \coordinate (b0\i) at (2.5cm,-\i*1cm);}
\foreach \i in {1,...,3} {  \node (b0\i) at (b0\i) { \nscale{$\top_{i-1,\i}$}}; }
\foreach \i in {4} {  \node (b0\i) at (b0\i) { $A_{i-1}$}; }
%%%%%%%%%%%%
\foreach \i in {1,...,4} {   \coordinate (b1\i) at (7.5cm,-\i*1cm);}
\foreach \i in {1,...,3} {  \node (b1\i) at (b1\i) { \nscale{$\top_{i,\i}$}}; }
\foreach \i in {4} {  \node (b1\i) at (b1\i) { $A_{i}$}; }
%%%%%%%%%%%%
\foreach \i in {1,...,4} {   \coordinate (b2\i) at (12.5cm,-\i*1cm);}
\foreach \i in {1,...,3} {  \node (b2\i) at (b2\i) { \nscale{$\top_{i+1,\i}$}}; }
\foreach \i in {4} {  \node (b2\i) at (b2\i) { $A_{i+1}$}; }

\graph {
(t0) <->["\escale{$\otimes_{i-1}$}"] (t1) ->["\escale{$\odot_{i-1}$}"] (t2)<->["\escale{$\odot_{i-1}$}"] (t3)<->[ "\escale{$\otimes_{i-1}$}"] (t4)<->["\escale{$\otimes_i$}"] (t5)->["\escale{$\odot_i$}"] (t6)<->["\escale{$\odot_i$}"] (t7)<->["\escale{$\otimes_i$}"] (t8)<->[ "\escale{$\otimes_{i+1}$}"] (t9)->["\escale{$\odot_{i+1}$}"] (t10)<->["\escale{$\odot_{i+1}$}"] (t11)<->[ "\escale{$\otimes_{i+1}$}"] (t12);
(t2) ->["\escale{$\ominus_{i-1}$}"] (b01)<->["\escale{$\ominus_{i-1}$}"] (b02)<->[ "\escale{$\odot_{i-1}$}"] (b03)<->[ "\escale{$\oplus_{i-1}$}"] (b04);
(t6) ->["\escale{$\ominus_{i}$}"] (b11)<->["\escale{$\ominus_{i}$}"] (b12)<->[ "\escale{$\odot_{i}$}"] (b13)<->[ "\escale{$\oplus_{i}$}"] (b14);
(t10) ->["\escale{$\ominus_{i+1}$}"] (b21)<->["\escale{$\ominus_{i+1}$}"] (b22)<->[ "\escale{$\odot_{i+1}$}"] (b23)<->[ "\escale{$\oplus_{i+1}$}"] (b24);
};
\end{scope}
\end{tikzpicture}
\caption{A snippet of the joining $A(U)$ that shows the definition of the transition function $\delta$ on the states $\bot_{4j}, \dots, \bot_{4(j+1)}, \top_{j,1},\dots, \top_{j,3}$ and events $\otimes_j,\odot_j,\ominus_j,\oplus_j$, where $j\in \{ i-1,i,i+1\}$.}
\label{fig:joining}
\end{figure}

Notice that $A(U)$ is obviously $2$-grade on $\top\cup\bot$, that is, all states of $\top\cup \bot$ have at most two incoming and two outgoing edges.
In particular, the state $\top_{i,3}$,  $i\in \{0,\dots, m-1\}$, has exactly two incoming and two outgoing edges labeled by $\odot_i$ and $\oplus_i$, respectively.
As a result, if $U$ is $g$-grade for some $g \geq 2$ and if the initial states of all TSs in $U$ have at most $g-1$ incoming and at most $g-1$ outgoing edges then $A(U)$ is $g$-grade, too. 
The following lemma certifies the validity of the joining operation for the unions and the types of nets that occur in our reduction.
By construction, these unions satisfy the requirements of the lemma:

\begin{lemma}\label{lem:union_validity}
Let $\tau\in \sigma_1\cup\sigma_2$ and let $U = U(A_0, \dots, A_n)$ be a union of TSs $A_0, \dots, A_n$ satisfying the following conditions:
\begin{enumerate}
\item
 For every event $e$ in $E(U)$ there is at least one state $s$ in $S(U)$ with $\neg (s \edge{e})$.
 \item
For all $i \in \{0, \dots, n\}$ there is exactly one outgoing and one incoming arc at the initial state $s_0^i$ of the TS $A_i$, both are labeled with the same event $u_i$ that occurs nowhere else in $U$.
 \end{enumerate}
 
$A(U)$ has the $\tau$-(E)SSP if and only if $A(U)$ has the $\tau$-(E)SSP.
\end{lemma}
\begin{proof}
\emph{If}:
Projecting a $\tau$-region separating $s$ and $s'$, respectively inhibiting $e$ at $s$, in $A(U)$ to the component TSs yields a $\tau$-region separating $s$ and $s'$, respectively inhibiting $e$ at $s$ in $U$.
Hence, the $\tau$-(E)SSP of $A(U)$ trivially implies the $\tau$-(E)SSP of $U$.

\emph{Only if}:
%In the following, let \textsf{test} be any interaction of $\tau \cap \{\used,\free\}$.
A $\tau$-region $R$ of $U$ separating $s$ and $s'$, respectively inhibiting $e$ at $s$, can be completed to become an equivalent $\tau$-region $R'$ of $A(U)$ by setting
\begin{align*}
sup_{R'}(s'') &= \begin{cases}
sup_R(s''), & \text{if } s'' \in S(U),\\
sup_R(s), & \text{otherwise, that is, } s'' \in \bot \cup \top \text{ and}
\end{cases}\\
sig_{R'}(e') &= \begin{cases}
sig_R(e'), & \text{if } e' \in E(U),\\
\nop, & \text{if } e'\in \odot\cup \otimes\cup \ominus\\ 
\nop, & \text{if } e' = \oplus_i \text{ and } sup_R(s^i_0) = sup_R(s), i \in \{0, \dots, n\}\\
\swap, & \text{if } e' = \odot_i \text{ and } sup_R(s^i_0) \not= sup_R(s) = 1, i \in \{0, \dots, n\},\\
\end{cases}
\end{align*}
Notice that a $\tau$-region $R'$ like this, which inherits the property of inhibiting $e$ at $s$ from $R$, do also inhibit $e$ at all connector states, since $sup_{R'}(s) = sup_{R'}(s'')$ for all $s''\in  \bot \cup \top$.
This has the following consequence: As every event $e\in E(U)$ has at least one state $s\in S(U)$ with $\neg s\edge{e}$, the ESSP of $U$ implies that $U$ has at least one inhibiting region $R$ for every event $e$.
Hence, for every event $e$ we can use the respective region to create $R'$ as defined above, inhibiting $e$ at every connector state of the TS $A(U)$.

For the (E)SSP of $A(U)$ it is subsequently sufficient to analyze (event) state separation concerning just the connector states and events.
By the assumption, we have to consider the cases $\tau\in \sigma_1$ and $\tau\in\sigma_2$, respectively.
Firstly, we let $\tau\in \sigma_1$ and present $\tau$-regions, compressed in accordance to Section~\ref{sec:compressed}, that prove the $\tau$-ESSP of $A(U)$.
Secondly, we do the same for $\tau\in \sigma_2$.
Finally, we show the $\tau$-SSP of $A(U)$ by regions whose signatures get along with the interactions $\nop,\set,\swap$, which are included by all $\tau\in \sigma_1\cup\sigma_2$.

Let $i\in \{0,\dots, n\}$.
In the following we present regions that together inhibit the events $\odot_i, \otimes_i,\ominus_i, \oplus_i$:

Region $R^{\odot_i}=(s^{\odot_i}, sig^{\odot_i})$ with $sup^{\odot_i}=S(A(U))\setminus \{\bot_{4i+1}, \bot_{4i+2}, \bot_{4i+3}, \top_{i,2}, \top_{i,3}\}$ inhibits $\odot_i$.

The following two regions $R^{\otimes_i}_1$ and $R^{\otimes_i}_2$ prove the inhibition of $\otimes_i$, where the former is dedicated to the states $\bot_{4i-3}$ (if present), $\bot_{4i+2}$ and $\bot_{4i+5}$ and the latter to the remaining states:

Region $R^{\otimes_i}_1=(sup^{\otimes_i}_1, sig^{\otimes_i}_1)$ is defined by $sup^{\otimes_i}_1=\{\bot_{4j+1}, \bot_{4j+2}, \bot_{4j+3}, \top_{j,1}, \top_{j,2},\top_{j,3 } \mid j\in \{0,\dots, i-1,i,\dots, n\}\} \cup \bigcup^n_{i\not=j=0}S(A_j)$.

Region $R^{\otimes_i}_2=(sup^{\otimes_i}_2, sig^{\otimes_i}_2)$ is defined by $sup^{\otimes_i}_2= S(A(U))\setminus \{\bot_{4i-3}, \bot_{4i},\dots, \bot_{4i+5}\}$.

The following region $R^{\oplus_i}_1$ inhibits $\oplus_i$ at $\bot_{4i+2}$ and the states $S(A_i)\setminus \{s^i_0\}$ and the region $R^{\oplus_i}_2$ inhibits $\oplus_i$ at the remaining states:

Region $R^{\oplus_i}_1=(sup^{\oplus_i}_1, sig^{\oplus_i}_1)$ is defined by $sup^{\oplus_i}_1=S(A(U))\setminus \{s^i_0, \top_{i,3}, \bot_{4i+2}\}$.
Notice that, by assumption of item two, the event $u_i$ occurs unambiguously like $s^i_0\fbedge{u_i}$ in $A_i$.

Region $R^{\oplus_i}_2=(sup^{\oplus_i}_2, sig^{\oplus_i}_2)$ is defined by $sup^{\oplus_i}_2=\{\top_{i,1} \top_{i,2}, \bot_{4i+2}\}$.

Finally, the event $\ominus_i$ is inhibited by the region $R^{\ominus_i}=(sup^{\ominus_i}, sig^{\ominus_i})$ which is defined by $sup^{\oplus_i}= S(A(U)) \setminus \{ \bot_{4i+2}, \top_{i,1},\top_{i,2} \}$.

By the arbitrariness of $i$, so far we have proven the $\tau$-ESSP of $A(U)$ for $\tau\in\sigma_1$.
We proceed by doing the same for $\tau\in\sigma_2$, where $\{\nop,\set,\swap,\used\}\subseteq \tau$.
To reuse a just introduced region $R=(sup, sig)$ we define its complement $\overline{R}=(\overline{sup}, \overline{sig}) $ as follows: $\overline{sup}=S(A(U))\setminus sup$.
Notice that we never build a complement of a region $R$ where $sig(e)=\set$ for some event $e$.

Let $i\in \{0,\dots, m-1\}$.
The event $\odot_i$ is inhibited by the complement of $R^{\odot_i}$.

The event $\otimes_i$ is inhibited at the states $\{ \bot_{4j+2}, \bot_{4j+3}, \top_{j,1},\top_{j,2},\top_{j,3}\} \cup S(A_j)$, where $0\leq j\in \{i-1,i+1\}$, and at $\top_{i,1}, \top_{i,2}$ by the complement of $R^{\otimes_i}_1$.

Moreover, $\otimes_i$ is inhibited at the remaining states by $R=(sup, sig	)$ which is defined by $sup=\{\bot_{4j+2}, \bot_{4j+3}, \top_{j,1},\top_{j,2},\top_{j,3}\} \cup S(A_j)\cup \{\bot_{4i}, \bot_{4i+1}, \bot_{4i+3},\bot_{4i+4}, \top_{i,1}, \top_{i,2}\}$, where $0\leq j\in \{i-1,i+1\}$.

The event $\oplus_i$ is inhibited by the complements of $R^{\oplus_i}_1$ and $R^{\oplus_i}_2$ and the event $\ominus_i$ by the complement of $R^{\ominus_i}$.

So far the $\tau$-ESSP of $A(U)$ for $\tau\in \sigma_1\cup\sigma_2$ is proven. 
Recall for the $\tau$-SSP that the states $S(A_i), i\in \{0,\dots,n\}$ are separated within $A_i$ and $sup=S(A_i)$ separates them in $A(U)$.
It remains to consider the states of $\top \cup\bot$.
We present separating regions that go along with the interactions $\{\nop, \set,\swap\}$ and, thus, are valid for $\tau\in \sigma_1\cup\sigma_2$.
To do so, we modify the concept of Section~\ref{sec:compressed}:
If $(sup, sig)$ is a region obtained by this concept then we get a corresponding region $(sup, sig')$ where we define $sig'(e)=sig(e)$ if $sig(e)\not\in \{\free,\used\}$ and, otherwise, $sig(e)=\nop$ for all $e\in E(A(U))$.

Let $i\in \{0,\dots, n\}$  and region $R_i=(sup_i, sig_i)$ be defined by $sup_i=S(A(U))\setminus (\{\bot_0,\dots, \bot_{4i+1}\}\cup\{\top_{j,1}, \top_{j,2},\top_{j,3} \mid j\in \{0,\dots, i-1\}\} \cup \bigcup^{n}_{j=0}S(A_j))$.
$R_i$ separates all states $s\in sup_i$ from all of $\{\bot_0,\dots, \bot_{4i+1}\}\cup\{\top_{j,1}, \top_{j,2},\top_{j,3} \mid j\in \{0,\dots, i-1\}\}$.
Thus, it only remains to show that the states $\bot_{4i+2}, \dots, \bot_{4i+5},\top_{i,1}, \top_{i,2}, \top_{i,3}$, where $i\in \{0,\dots, n\}$, are pairwise separable.
The already introduced regions (replacing \free\ by \nop) are sufficient, to be seen by the following listing:

The state $\top_{i,1}$ is separated by $R^{\odot_i}$ from $\bot_{4i+2}, \bot_{4i+3}, \top_{i,2}$ and $\top_{i,3}$ and by $R^{\otimes_i}_2$ from $\bot_{4i+4}$ and $\bot_{4i+5}$; 
state $\top_{i,2}$ is separated by $R^{\otimes_i}_2$ from $\bot_{4i+2}, \dots, \bot_{4i+5}$ and by $R^{\oplus_i}_2$ from $\top_{i,2}$; 
state $\top_{i,3}$ is separated by $R^{\otimes_i}_2$ from $\bot_{4i+2}, \dots, \bot_{4i+5}$; 
state $\bot_{4i+2}$ is separated by $R^{\odot_i}$ from $\bot_{4i+4}$ and $\bot_{4i+5}$ and by $R^{\otimes_i}_1$ from $\bot_{4i+3}$; 
state $\bot_{4i+3}$ is separated by $R^{\odot_i}$ from $\bot_{4i+4}$ and $\bot_{4i+5}$; 
state $\bot_{4i+4}$ is separated by $R^{\otimes_i}_1$ from $\bot_{4i+5}$.
\end{proof}

%\vspace*{-1cm}%%%%%%%%%%%%%%%%%%%%%%%%%%%%%%%%%%%%%%%%%
\subsection{Manual for the Proof of Theorem~\ref{the:main_result}}\label{sec:proof_main_result}%
%%%%%%%%%%%%%%%%%%%%%%%%%%%%%%%%%%%%%%%%%%%%%%%%%

The input to our approach is a set $\sigma \in \{\sigma_1,\sigma_2\}$ and a cubic monotone boolean $3$-CNF $\varphi = \{\zeta_0, \dots, \zeta_{m-1}\}$, a set of negation-free $3$-clauses over the variables $V(\varphi)$ such that every variable is a member of exactly three clauses.
According to~\cite{DBLP:journals/dcg/MooreR01}, it is NP-complete to decide if $\varphi$ has a one-in-three model, that is, a subset $M \subseteq V(\varphi)$ of variables that hit every clause exactly once: $\vert M \cap \zeta_i\vert  = 1$ for all $i \in \{0, \dots, m-1\}$.
The result of the reduction is a $2$-grade union $U^\sigma_\varphi$ of gadget TSs which satisfies the following condition:

\newtheorem{condition}{Condition}
\begin{mdframed}[backgroundcolor=green!10] 
\begin{condition}\label{con:manual}
\begin{enumerate}
\item\label{con:one}
There is event $k \in E(U^\sigma_\varphi)$ and state $h_{0,2} \in S(U^\sigma_\varphi)$ such that $\neg h_{0,2}\edge{k}$, that is, $k$ has to be inhibited at $h_{0,2}$.
\item\label{con:two}
There are events $V=\{v_0,\dots,v_{4m-1}\}\subseteq  E(U^\sigma_\varphi)$ and $W=\{w_0,\dots,w_{m-1}\}\subseteq E(U^\sigma_\varphi)$ and $Acc=\{a_0,\dots, a_{3m-1}\} \subseteq E(U^\sigma_\varphi) $.
\item\label{con:three}
If $(sup, sig)$ is a $\tau$-region, where $\tau \in \sigma$, that inhibits $k$ at $h_{0,2}$ then the following is true:
	\begin{enumerate}
	\item
	$V\subseteq sig^{-1}(\swap)$ and $W\cap sig^{-1}(\swap)=\emptyset$ and $Acc\cap sig^{-1}(\swap)=\emptyset$,
	\item
	$sig(k)=\free$ and $sup(h_{0,2})=1$ or $sig(k)=\used$ and $sup(h_{0,2})=0$.
	\end{enumerate}
\item\label{con:four}
The variables $V(\varphi)$ are a subset of $E(U^\sigma_\varphi)$, the union events.
\item\label{con:five}
If $(sup, sig )$ is a region of $U^\sigma_\varphi$ satisfying Condition~\ref{con:manual}.\ref{con:three} then $M=\{X\in V(\varphi) \mid sig(X) =\set\}$ or $M'=\{X\in V(\varphi) \mid sig(X)=\res\}$ is a one-in-three model of $\varphi$.
\item\label{con:six}
If $\varphi$ has a one-in-three model $M$ and $\tau\in \sigma$ then $U^\sigma_\varphi$ has the $\tau$-ESSP and the $\tau$-SSP.
\end{enumerate}
\end{condition}
\end{mdframed}
The precise meanings of the different items of Condition~\ref{con:manual} are developed in the following sections.
However, we can already justify that a polynomial time reduction that yields a union which satisfies Condition~\ref{con:manual} proves Theorem~\ref{the:main_result} as the following implications are true:
\noindent
\begin{mdframed}[backgroundcolor=green!10] 
$\varphi$ is one-in-three satisfiable $\stackrel{\text{6.}}{\implies}$ 
$U^\sigma_\varphi$ has the $\tau$-ESSP \& $\tau$-SSP $\stackrel{\text{def.}}{\implies}$   
$k$ is $\tau$-inhibitable at $h_{0,6}$ in $U^\sigma_\varphi$ $\stackrel{\text{3./5.}}{\implies}$ $\varphi$ is one-in-three satisfiable.
\end{mdframed}
Especially, $\varphi$ is one-in-three satisfiable if and only if $U^\sigma_\varphi$ is $\tau$-feasible, that is, it has the $\tau$-ESSP and the $\tau$-SSP. 
By Lemma~\ref{lem:union_validity}, this proves NP-hardness of $\tau$-feasibility for all $\tau\in \sigma_1\cup \sigma_2$.
Secondly, every remaining type $\tilde{\tau}$ of Theorem~\ref{the:main_result} is isomorphic to one of the already covered cases $\tau$.
Hence, by Lemma~\ref{lem:isomorphic_types}, this also proves NP-hardness of $\tilde{\tau}$-feasibility which, by feasibility being in NP, justifies Theorem~\ref{the:main_result}. 

The following sections are dedicated to the introduction of $U^\sigma_\varphi$ and to the proof of Condition~\ref{con:manual}.
In these sections, we often refer to the statements of the following simple observation:
\begin{observation}[Without proof.]\label{obs:basics}
Let $A$ be a TS, $\tau$ a boolean type of net and $(sup, sig)$ a $\tau$-region of $A$.
If $s\fbedge{e}s'$ are transitions of $A$ then $sup(s)\not=sup(s')$ if and only if $sig(e)=\swap$.
If $s\edge{e}s'\fbedge{e}s''$ are transitions of $A$, where $s,s',s''$ are pairwise distinct, and $sig(e)=\swap$ then $sup(s)=sup(s'')$.
Moreover, if $P=s\edge{e_1}\dots\edge{e_n}s'$ is a path of $A$ then the image $sup(s)\edge{sig(e_1)}\dots \edge{sig(e_n)}sup(s')$ of $P$ under $R$ is a path in $\tau$.
\end{observation}

%%%%%%%%%%%%%%%%%%%%%%%%%%%%%%%%%%%%%%%%%%%%%%%%%%%%%%
\subsection{Details for Condition~\ref{con:manual}.\ref{con:one}-Condition~\ref{con:manual}.\ref{con:three}}%
%%%%%%%%%%%%%%%%%%%%%%%%%%%%%%%%%%%%%%%%%%%%%%%%%%%%%%

In this section, we let $\sigma\in \{\sigma_1,\sigma_2\}$  and introduce the gadgets $U^\sigma_\varphi$ that are relevant for Condition~\ref{con:manual}.\ref{con:one}-Condition~\ref{con:manual}.\ref{con:three}.

For a start, the union $U^\sigma_\varphi$ has for every $j\in \{0,\dots, 4m-1\}$ the following TS $H_{j}$:
\begin{center}
\begin{tikzpicture}[new set = import nodes]
\begin{scope}[nodes={set=import nodes}]% make all nodes part of this set
 	%head H_j for sigma_4/sigma_5
		\node (B) at (-0.9,0) {$H_j=$};
		\foreach \i in {0,...,5} { \coordinate (\i) at (\i*1.5cm,0) ;}
		\foreach \i in {2,5} {, rounded corners] (\i) +(-0.5,-0.25) rectangle +(0.5,0.35);}
		\foreach \i in {0,...,5} { \node (\i) at (\i*1.5cm,0) {\scalebox{\nodeScale}{$h_{j,\i}$}};}
\graph {
	(import nodes);
			0 <->["\scalebox{\edgeScale}{$k$}"]1<->["\scalebox{\edgeScale}{$m$}"]2 <->["\scalebox{\edgeScale}{$v_j$}"]3 <->["\scalebox{\edgeScale}{$k$}"]4  <->["\scalebox{\edgeScale}{$\_$}"]5;

			};
\end{scope}
\end{tikzpicture}
\end{center}
The TS $H_j$ has the initial state $h_{j,5}$ and the events are $k,v_j,$ and $m$.
While every $H_j$ provides a new $v_j$, the events $k$ and $m$ are applied by every TS $H_0,\dots, H_{4m-1}$.
Thus, the TSs $H_0,\dots, H_{4m-1}$ together provide the events of $V=\{v_0,\dots, v_{4m-1}\}$.
Moreover, $H_j$ has a unique event, say $u_j$, that is sketched by the underscore and occurs nothing elsewhere in $U^\sigma_\varphi$.
The only purpose of this event is to satisfy the requirements of Lemma~\ref{lem:union_validity} and it has no impact on the proof of Condition~\ref{con:manual}. 
The TS $H_0$ provides the state $h_{0,2}$ and by $\neg h_{0,2}\edge{k}$ the event $k$ is to inhibit at this state.
We argue that a region $(sup, sig)$ that inhibits $k$ at $h_{0,2}$ satisfies $V\subseteq sig^{-1}(\swap)$ and either $sig(k)=\free$ and $sup(h_{0,2})=1$ or $sig(k)=\used$ and $sup(h_{0,2})=0$.
If $(sup, sig)$ is a $\tau$-region that inhibits $k$ at $h_{0,2}$ then, by definition, the interaction $sig(k)$ does not occur at $ sup(h_{0,2})$ in $\tau$.
Thus, by definition of $\sigma_1$ and $\sigma_2$, we have $sig(k)\in \{\free, \used\}$ for all considered types $\tau$.
Let's discuss the case $sig(k)=\free$, which immediately implies $sup(h_{0,2})=1$.
By $sig(k)=\free$ and $\edge{k}h_{j,1}$ and $\edge{k}h_{j,3}$  we get $sup(h_{j,1})=sup(h_{j,3})=0$ for all $j\in \{0,\dots, 4m-1\}$.
By $sup(h_{0,1})=0$ and $sup(h_{0,2})=1$ and Observation~\ref{obs:basics} we also get $sig(m)=\swap$.
Thus, by $sup(h_{j,1})=0$ and $h_{j,1}\fbedge{m}h_{j,2}$ we get $sup(h_{j,2})=1$ for all $j\in \{0,\dots, 4m-1\}$.
Finally, $sup(h_{j,2})=1$ and $sup(h_{j,3})=0$ and Observation~\ref{obs:basics} imply $V\subseteq sig^{-1}(\swap)$.
Symmetrically, one argues that $sig(k)=\used$ implies $sup(h_{0,2})=0$ and $V\subseteq sig^{-1}(\swap)$.

The union $U^\sigma_\varphi$ has also the following TS $F_0, F_1$ and $F_2$: 
\begin{center}
\begin{tikzpicture}[new set = import nodes]
\begin{scope}[nodes={set=import nodes}]% make all nodes part of this set
 	%generate helper q_0
		\node (F_0) at (-0.75,0) {$F_0=$};
		\foreach \i in {0,...,8} { \coordinate(\i) at (\i*1.55cm,0);}
		\foreach \i in {2,5,8} {, rounded corners] (\i) +(-0.5,-0.25) rectangle +(0.5,0.3);}
		\foreach \i in {0,...,8} { \node (\i) at (\i*1.55cm,0) {\scalebox{\nodeScale}{$f_{0,\i}$}};}
\graph {
	(import nodes);
			0 <->["\scalebox{\edgeScale}{$k$}"]1<->["\scalebox{\edgeScale}{$m$}"]2 <->["\scalebox{\edgeScale}{$q_0$}"]3 <->["\scalebox{\edgeScale}{$k$}"]4  <->["\scalebox{\edgeScale}{$m$}"]5<->["\scalebox{\edgeScale}{$q_1$}"]6<->["\scalebox{\edgeScale}{$k$}"]7<->["\scalebox{\edgeScale}{$\_$}"]8;
		};
\end{scope}	
\begin{scope}[nodes={set=import nodes}, yshift=-1cm]% make all nodes part of this set
		\node (F_1) at (-0.75,0) {$F_1=$};
		\foreach \i in {0,...,5} { \coordinate (\i) at (\i*1.6cm,0) ;}
		\foreach \i in {2,5} {, rounded corners] (\i) +(-0.5,-0.25) rectangle +(0.5,0.3);}
		\foreach \i in {0,...,5} { \node (\i) at (\i*1.6cm,0) {\scalebox{\nodeScale}{$f_{1,\i}$}};}

\graph {
	(import nodes);
			0 <->["\scalebox{\edgeScale}{$k$}"]1<->["\scalebox{\edgeScale}{$q_2$}"]2 <->["\scalebox{\edgeScale}{$q_3$}"]3 <->["\scalebox{\edgeScale}{$k$}"]4  <->["\scalebox{\edgeScale}{$\_$}"]5;
			};
\end{scope}
\begin{scope}[nodes={set=import nodes}, yshift=-2cm]% make all nodes part of this set
 	%generate z
		\node (F_2) at (-0.75,0) {$F_2=$};
		\foreach \i in {0,...,9} { \coordinate (\i) at (\i*1.4cm,0);}
		\foreach \i in {2} {, rounded corners] (\i) +(-0.5,-0.25) rectangle +(0.5,0.3);}
		\foreach \i in {6} {, rounded corners] (\i) +(-1.9,-0.25) rectangle +(0.5,0.3);}
		\foreach \i in {9} {, rounded corners] (\i) +(-0.5,-0.25) rectangle +(0.3,0.3);}
		\foreach \i in {0,...,9} { \node (\i) at (\i*1.4cm,0) {\scalebox{\nodeScale}{$f_{2,\i}$}};}
\graph {
	(import nodes);
			0 <->["\scalebox{\edgeScale}{$k$}"]1<->["\scalebox{\edgeScale}{$q_2$}"]2 <->["\scalebox{\edgeScale}{$q_0$}"]3 <->["\scalebox{\edgeScale}{$z$}"]4  <->["\scalebox{\edgeScale}{$q_1$}"]5 <->["\scalebox{\edgeScale}{$z$}"]6 <->["\scalebox{\edgeScale}{$q_3$}"]7 <->["\scalebox{\edgeScale}{$k$}"]8<->["\scalebox{\edgeScale}{$\_$}"]9;
			};
\end{scope}
\end{tikzpicture}
\end{center}
The TSs $F_0,F_1$ and $F_2$ have the initial states $f_{0,8}, f_{1,5}$ and $f_{2,9}$ and apply the events $k,m,q_0,\dots, q_3$ and $z$, respectively.
Again, the underscores represent unique events that do not occur elsewhere in $U^\sigma_\varphi$ and have the only purpose to satisfy the requirements of Lemma~\ref{lem:union_validity}.
The TSs $F_0,F_1$ and $F_2$ have exactly the following task that concerns the event $z$: 
If $(sup, sig)$ is a region of $U^\sigma_\varphi$ that inhibits $k$ at $h_{0,2}$ then $F_0,F_1$ and $F_2$ use the signatures of $k$ and $m$ to ensure $sig(z)\in \{ \nop,\swap\}$.
This property of $z$ is used by $U^\sigma_\varphi$ to ensure a certain signature for other events.
More exactly, the union $U^\sigma_\varphi$ has for every $j\in \{0,\dots, m-1\}$ and for every $\ell\in \{0,\dots, 3m-1\}$ the following TSs $G_j$ and $D_\ell$, which together provide the sets $W=\{w_0,\dots, w_{m-1}\}$ and $Acc=\{a_0,\dots, a_{3m-1}\}$ and apply $k$ and $z$ to ensure $W\cap sig^{-1}(\swap)=\emptyset$ and $Acc\cap sig^{-1}(\swap)=\emptyset$ for every region $(sup, sig)$ that inhibits $k$ at $h_{0,2}$:
\begin{center}
\begin{tikzpicture}[new set = import nodes]
\begin{scope}[nodes={set=import nodes}]% make all nodes part of this set
 	%wire generator for sigma_4/sigma_5
		\node (B) at (-0.75,0) {$G_j=$};
		\foreach \i in {0,...,8} { \coordinate (\i) at (\i*1.5cm,0) ;}
		\foreach \i in {4} {, rounded corners] (\i) +(-3.5,-0.25) rectangle +(0.5,0.35);}
		\foreach \i in {8} {, rounded corners] (\i) +(-0.5,-0.25) rectangle +(0.5,0.35);}
		\foreach \i in {0,...,8} { \node (g\i) at (\i) {\scalebox{\nodeScale}{$g_{j,\i}$}};}
\graph {
	(import nodes);
			g0 <->["\scalebox{\edgeScale}{$k$}"]g1<->["\scalebox{\edgeScale}{$y_j$}"]g2- >["\scalebox{\edgeScale}{$z$}"]g3 <->["\scalebox{\edgeScale}{$z$}"]g4 <->["\scalebox{\edgeScale}{$y_j$}"]g5<->["\scalebox{\edgeScale}{$w_j$}"]g6<->["\scalebox{\edgeScale}{$k$}"]g7<->["\scalebox{\edgeScale}{$\_$}"]g8;
			};	
\end{scope}
\begin{scope}[nodes={set=import nodes}, yshift =-1cm]% make all nodes part of this set
 	%accordance duplicator for sigma_4/sigma_5
		\node (B) at (-0.75,0) {$D_\ell=$};
		\foreach \i in {0,...,8} { \coordinate (\i) at (\i*1.5cm,0) ;}
		\foreach \i in {4} {, rounded corners] (\i) +(-3.5,-0.25) rectangle +(0.5,0.35);}
		\foreach \i in {8} {, rounded corners] (\i) +(-0.5,-0.25) rectangle +(0.5,0.35);}
		\foreach \i in {0,...,8} { \node (d\i) at (\i) {\scalebox{\nodeScale}{$d_{\ell,\i}$}};}
\graph {
	(import nodes);
			d0 <->["\scalebox{\edgeScale}{$k$}"]d1<->["\scalebox{\edgeScale}{$p_\ell$}"]d2- >["\scalebox{\edgeScale}{$z$}"]d3 <->["\scalebox{\edgeScale}{$z$}"]d4 <->["\scalebox{\edgeScale}{$p_\ell$}"]d5<->["\scalebox{\edgeScale}{$a_\ell$}"]d6<->["\scalebox{\edgeScale}{$k$}"]d7<->["\scalebox{\edgeScale}{$\_$}"]d8;
			};
\end{scope}
\end{tikzpicture}
\end{center}
In the remainder of this section, we first prove the announced functionality of $F_0,F_1$ and $F_2$ and then do so for $G_j$ and $D_\ell$.
If $(sup, sig)$ is a region that inhibits $k$ at $h_{0,2}$ then $sig(k)\in \{\free, \used\}$ and, by the former discussion, $sig(m)=\swap$.
By $sig(k)\in \{\free, \used\}$ and $sig(m)=\swap$ we get $sup(f_{0,1})=sup(f_{0,3})=sup(f_{0,4})=sup(f_{0,6})\not=sup(f_{0,2})=sup(f_{0,5})$.
This implies, by Observation~\ref{obs:basics}, that $sig(q_0)=sig(q_1)=\swap$.
Moreover, again by $sig(k)\in \{\free, \used\}$, we have $sup(f_{1,1})=sup(f_{1,3})$, which implies $sig(q_2)=\swap$ if and only if $sig(q_3)=\swap$ if and only if $sup(f_{1,1})\not=sup(f_{1,2})$.
This influences directly the states of $F_2$:
By $sig(k)\in \{\free, \used\}$ we have $sup(f_{2,1})=sup(f_{2,7})$, which, by Observation~\ref{obs:basics} and the behavior of $q_2,q_3$, implies that $sup(f_{2,2})=sup(f_{2,6})$.
This implies that if $sig(z)\not=\swap$ then $sig(z)=\nop$:
If $sig(z)\not=\swap$ then, by Observation~\ref{obs:basics}, we have that $sup(f_{2,3})=sup(f_{2,4})$ and $sup(f_{2,5})=sup(f_{2,6})$.
By $sig(q_1)=\swap$, we get $sup(f_{2,4})\not=sup(f_{2,5})$.
Thus, we have $0\edge{sig(z)}0$ and $1\edge{sig(z)}1$ in $\tau$ which implies $sig(z)=\nop$.
Altogether, this implies $sig(z)\in \{\nop,\swap\}$ for every region $(sup, sig)$ that inhibits $k$ at $h_{0,2}$.

It remains to argue that such a region implies $W\cap sig^{-1}(\swap)=\emptyset$ and $Acc\cap sig^{-1}(\swap)=\emptyset$:
Let $j\in \{0,\dots, m-1\}$.
By Observation~\ref{obs:basics} and $sig(z)\in \{\nop,\swap\}$, we have that $sup(g_{j,2})=sup(g_{j,4})$ and, by $sig(k)\in \{\free, \used\}$, we get $sup(g_{j,1})=sup(g_{j,6})$.
Again by Observation~\ref{obs:basics}, the following is true:
If $sig(y_j)=\swap$ then $sup(g_{j,1}) \not= sup(g_{j,2})=sup(g_{j,4})\not=sup(g_{j,5})$, that is, $sup(g_{j,1}) = sup(g_{j,5}) =sup(g_{j,6})$.
This implies $sig(w_j)\not=\swap$.
If, otherwise, $sig(y_j)\not=\swap$ then $sup(g_{j,1}) = sup(g_{j,2})=sup(g_{j,4})=sup(g_{j,5})=sup(g_{j,6})$, which also implies $sig(w_j)\not=\swap$.
By the arbitrariness of $j$, this proves $W\cap sig^{-1}(\swap)=\emptyset$.
Moreover, as $G_j$ and $D_\ell$ are obviously isomorphic we have also proven that $Acc \cap sig^{-1}(\swap)=\emptyset$.

So far we have shown that $U^\sigma_\varphi$ satisfies the Conditions~\ref{con:manual}.\ref{con:one}-Condition~\ref{con:manual}.\ref{con:three}.
The following section is dedicated to the proof of Condition~\ref{con:manual}.\ref{con:four} and Condition~\ref{con:manual}.\ref{con:five}.

%%%%%%%%%%%%%%%%%%%%%%%%%%%%%%%%%%%%%%%%%%%%%%%%%%%%%%%%
\subsection{Details for Condition~\ref{con:manual}.\ref{con:four} and Condition~\ref{con:manual}.\ref{con:five}}%
%%%%%%%%%%%%%%%%%%%%%%%%%%%%%%%%%%%%%%%%%%%%%%%%%%%%%%%%

In this section, we introduce the remaining gadgets of $U^\sigma_\varphi$.
While the so far introduced gadgets are valid for both $\sigma_1$ and $\sigma_2$, this is no longer true for the gadgets which we introduce in this section.
The reason is that there are types in $\sigma_2$ which in a certain way are essentially different from the type of $\sigma_1$.
This requires a different construction to be able to satisfy Condition~\ref{con:manual}.\ref{con:six}.
In the following we first introduce the remaining gadgets for $\sigma_1$ and then the ones for $\sigma_2$.

The union $U^{\sigma_1}_\varphi$ has for clause $\zeta_i=\{X_{i,0}, X_{i,1}, X_{i,2}\}$, $i\in \{0,\dots, m-1\}$, the following TS $T_{i,0}$:
\begin{center}
\begin{tikzpicture}[new set = import nodes]
\begin{scope}[nodes={set=import nodes}]
\node (T_0) at (-0.9,0) {$T_{i,0}=$}; 
\foreach \i in {0,...,9} {\coordinate (\i) at (\i*1.5cm,0);}
\foreach \i in {10,...,17} { \pgfmathparse{int{\i-10}}   \coordinate (\i) at (13.5cm-\pgfmathresult*1.5cm,-1.25cm);}
\foreach \i in {0,...,17} {\node (t\i) at (\i) {\nscale{$t_{i,0,\i}$}}; }
\graph {
(t0) <->["\escale{$k$}"] (t1)<->["\escale{$v_{4i}$}"] (t2)<->["\escale{$a_{3i}$}"] (t3)<->[ "\escale{$X_{i,0}$}"] (t4)<-[ "\escale{$X_{i,0}$}"] (t5)<->["\escale{$a_{3i}$}"] (t6)<->["\escale{$a_{3i+1}$}"] (t7)<->[ "\escale{$X_{i,1}$}"] (t8)<-[ "\escale{$X_{i,1}$}"] (t9)<->["\escale{$a_{3i+1}$}"] (t10)<->["\escale{$a_{3i+2}$}"] (t11)<->[  "\escale{$X_{i,2}$}"] (t12)<-["\escale{$X_{i,2}$}"] (t13)<->["\escale{$a_{3i+2}$}"] (t14)<->["\escale{$w_i$}"] (t15)<->["\escale{$k$}"] (t16)<->["\escale{$\_$}"] (t17);
};
\end{scope}
\end{tikzpicture}
\end{center}
Additionally, the union $U^{\sigma_1}_\varphi$ has for clause $\zeta_i$ the following three TSs $T_{i,1}, T_{i,2}$ and $T_{i,3}$: 
\begin{center}
\begin{tikzpicture}[new set = import nodes]
\begin{scope}[nodes={set=import nodes}]% make all nodes part of this set
 	%head H_j for sigma_4/sigma_5
		\node (T_3) at (-0.9,0) {$T_{i,1}=$};
		\foreach \i in {0,...,4} { \coordinate (\i) at (\i*1.7cm,0) ;}
		\foreach \i in {0,...,4} { \node (\i) at (\i) {\nscale{$t_{i,1,\i}$}};}
\graph {
	(import nodes);
			0 <->["\escale{$X_{i,0}$}"]1<->["\escale{$v_{4i+1}$}"]2 <->["\escale{$X_{i,1}$}"]3 <->["\escale{$\_$}"]4;

			};
\end{scope}
\begin{scope}[nodes={set=import nodes}, yshift=-1cm]% make all nodes part of this set
 	%head H_j for sigma_4/sigma_5
		\node (B) at (-0.9,0) {$T_{i,2}=$};
		\foreach \i in {0,...,4} { \coordinate (\i) at (\i*1.7cm,0) ;}
		%\foreach \i in {2,5} {, rounded corners] (\i) +(-0.5,-0.25) rectangle +(0.5,0.35);}
		\foreach \i in {0,...,4} { \node (\i) at (\i) {\nscale{$t_{i,2,\i}$}};}
\graph {
	(import nodes);
			0 <->["\escale{$X_{i,0}$}"]1<->["\escale{$v_{4i+2}$}"]2 <->["\escale{$X_{i,2}$}"]3 <->["\escale{$\_$}"]4;

			};
\end{scope}
\begin{scope}[nodes={set=import nodes}, yshift=-2cm]% make all nodes part of this set
 	%head H_j for sigma_4/sigma_5
		\node (B) at (-0.9,0) {$T_{i,3}=$};
		\foreach \i in {0,...,4} { \coordinate (\i) at (\i*1.7cm,0) ;}
		%\foreach \i in {2,5} {, rounded corners] (\i) +(-0.5,-0.25) rectangle +(0.5,0.35);}
		\foreach \i in {0,...,4} { \node (\i) at (\i) {\nscale{$t_{i,3,\i}$}};}
\graph {
	(import nodes);
			0 <->["\escale{$X_{i,1}$}"]1<->["\escale{$v_{4i+3}$}"]2 <->["\escale{$X_{i,2}$}"]3 <->["\escale{$\_$}"]4;

			};
\end{scope}
\end{tikzpicture}
\end{center}
The TSs $T_{i,0}, \dots, T_{i,3}$ have the initial states $t_{i,0,17}, t_{i,1,4}, t_{i,2,4}$ and $t_{i,3,4}$, respectively, and use the variables $X_{i,0}, X_{i,1}$ and $X_{i,2}$ of $\zeta_i$ as events.
Recall that for Condition~\ref{con:manual}.\ref{con:five}, by $\res\not\in \tau$ for $\tau\in \sigma_1$, the set $M=\{ X\in V(\varphi) \mid sig(X)=\set\}$ has to be a one-in-three model for every region $(sup, sig)$ that inhibits $k$ at $h_{0,2}$.
To prove this condition, we let $i\in \{0,\dots, m-1\}$ be arbitrary and show that a corresponding region satisfies that there is an event $X\in \{X_{i,0}, X_{i,1}, X_{i,2}\}$ such that $sig(X)=\set$ and that $sig(Y)\not=\set$ for $Y\in \{X_{i,0}, X_{i,1}, X_{i,2}\}\setminus \{X\}$.
By the arbitrariness of $i$, this is simultaneously true for all clauses $\zeta_0,\dots, \zeta_{m-1}$.
Thus, $(sup, sig)$ selects exactly one variable of every clause via the \set-signature of the corresponding variable events, which actually makes $M=\{X\in V(\varphi)\mid sig(X)=\set\}$ a one-in-three model of $\varphi$.

If $\tau\in \sigma_1$ and if $(sup, sig)$ is a $\tau$-region of $U^{\sigma_1}_\varphi$ that inhibits $k$ at $h_{0,2}$ then, by definition of $\sigma_1$ and the discussions above, we have that $sig(k)=\free$.
By $sig(k)=\free$, we have $sup(t_{i,0,1})=sup(t_{i,0,15})=0$.
Furthermore, Condition~\ref{con:manual}.\ref{con:three} implies $V \subseteq sig^{-1}(\swap)$ and $W\cap sig^{-1}(\swap)=\emptyset$ and $Acc \cap sig^{-1}(\swap)=\emptyset$.
Thus, by $sup(t_{i,0,1})=0$ and $V \subseteq sig^{-1}(\swap)$ we get $sup(t_{i,0,2})=1$, which with $Acc \cap sig^{-1}(\swap)=\emptyset$ and Observation~\ref{obs:basics} implies $sup(t_{i,0,3})=1$.
Similarly, $sup(t_{i,0,15})=0$, $W\cap sig^{-1}(\swap)=\emptyset$, $Acc \cap sig^{-1}(\swap)=\emptyset$ and Observation~\ref{obs:basics} imply $sup(t_{i,0,13})=0$.
As a result, the image of the following sequence $P_i$ of $T_{i,0}$ contains a subsequence in $\tau$ which starts $0$ and terminate at $1$.
The sequence $P_i$ is given by $P_i=$
\begin{center}
\begin{tikzpicture}[new set = import nodes]
\begin{scope}[nodes={set=import nodes}]
\foreach \i in {3,...,13} { \pgfmathparse{int{\i-3}}   \coordinate (\i) at (\pgfmathresult*1.5cm,0);}
\foreach \i in {3,...,13} {\node (t\i) at (\i) {\nscale{$t_{i,0,\i}$}}; }
\graph {
 (t3)<->[ "\escale{$X_{i,0}$}"] (t4)<-["\escale{$X_{i,0}$}"] (t5)<->["\escale{$a_{3i}$}"] (t6)<->["\escale{$a_{3i+1}$}"] (t7)<->["\escale{$X_{i,1}$}"] (t8)<-[ "\escale{$X_{i,1}$}"] (t9)<->["\escale{$a_{3i+1}$}"] (t10)<->["\escale{$a_{3i+2}$}"] (t11)<->[ "\escale{$X_{i,2}$}"] (t12)<-["\escale{$X_{i,2}$}"] (t13);
};
\end{scope}
\end{tikzpicture}
\end{center}
We obtain immediately that there is at least one event (of $Pi$) whose signature realizes the state change from $0$ to $1$ in $\tau$.
By $Acc\cap sig^{-1}(\swap)=\emptyset$, this event can not be any of $a_{3i},a_{3i+1}, a_{3i+2}$.
Moreover, if $X\in \{X_{i,0}, X_{i,1}, X_{i,2}\}$, $sig(X)\in \{ \nop, \swap, \free\}$ then for $s\edge{X}s'\fbedge{X}s''$, where $s,s',s''$ are pairwise distinct, we get $sup(s)=sup(s'')$.
Thus, there has to be an event $X\in \{X_{i,0}, X_{i,1}, X_{i,2}\}$ such that $sig(X)=\set$. 
In the following, we argue that if $X\in \{X_{i,0}, X_{i,1}, X_{i,2}\}$ and $sig(X)=\set$ then $sig(Y)\not=\set$ for $Y\in \{X_{i,0}, X_{i,1}, X_{i,2}\}\setminus \{X\}$.

If $sig(X_{i,0})=\set$ then, by $\edge{X_{i,0}}t_{i,1,1}$ and $\edge{X_{i,0}}t_{i,2,1}$, we get that $sup(t_{i,1,1})=sup(t_{i,2,1})=1$.
Moreover, by $t_{i,1,1}\edge{v_{4i+1}}t_{i,1,2}$, $t_{i,2,1}\edge{v_{4i+2}}t_{i,2,2}$, $V\subseteq sig^{-1}(\swap)$ and Observation~\ref{obs:basics}, we conclude $sup(t_{i,1,2})=sup(t_{i,2,2})=0$.
Thus, by $\edge{X_{i,1}}t_{i,1,2}$ and $\edge{X_{i,2}}t_{i,2,2}$, we obtain that $sig(X_{i,1})\not=\set$ and $sig(X_{i,2})\not=\set$.
By the symmetry of $T_{i,1}, T_{i,2}$ and $T_{i,3}$ (and similar arguments) it is easy to see that $sig(X_{i,1})=\set$ implies $sig(X_{i,0})\not=\set$ and $sig(X_{i,2})\not=\set$ and that $sig(X_{i,2})=\set$ implies $sig(X_{i,0})\not=\set$ and $sig(X_{i,1})\not=\set$. 

Altogether, so far we have proven that $U^{\sigma_1}_\varphi$ satisfies Condition~\ref{con:manual}.\ref{con:four} and Condition~\ref{con:manual}.\ref{con:five}.
In the remainder of this section we argue that $U^{\sigma_2}_\varphi$ does too.

The union $U^{\sigma_2}_\varphi$ has for clause $\zeta_i=\{X_{i,0}, X_{i,1}, X_{i,2}\}$, $i\in \{0,\dots, m-1\}$, the following TS $T'_{i,0}$:
\begin{center}
\begin{tikzpicture}[new set = import nodes]
\begin{scope}[nodes={set=import nodes}]
\node (T_0) at (-0.9,0) {$T_{i,0}=$}; 
\foreach \i in {0,...,9} {\coordinate (\i) at (\i*1.5cm,0);}
\foreach \i in {10,...,17} { \pgfmathparse{int{\i-10}}   \coordinate (\i) at (13.5cm-\pgfmathresult*1.5cm,-1.25cm);}
\foreach \i in {0,...,17} {\node (t\i) at (\i) {\nscale{$t_{i,0,\i}$}}; }
\graph {
(t0) <->["\escale{$k$}"] (t1)<->["\escale{$w_i$}"] (t2)<->["\escale{$a_{3i}$}"] (t3)<->[ "\escale{$X_{i,0}$}"] (t4)<-[ "\escale{$X_{i,0}$}"] (t5)<->["\escale{$a_{3i}$}"] (t6)<->["\escale{$a_{3i+1}$}"] (t7)<->[ "\escale{$X_{i,1}$}"] (t8)<-[ "\escale{$X_{i,1}$}"] (t9)<->["\escale{$a_{3i+1}$}"] (t10)<->["\escale{$a_{3i+2}$}"] (t11)<->[  "\escale{$X_{i,2}$}"] (t12)<-["\escale{$X_{i,2}$}"] (t13)<->["\escale{$a_{3i+2}$}"] (t14)<->["\escale{$v_{4i}$}"] (t15)<->["\escale{$k$}"] (t16)<->["\escale{$\_$}"] (t17);
};
\end{scope}
\end{tikzpicture}
\end{center}
Notice that the (only) difference between $T'_{i,0}$ and $T_{i,0}$ is the switched position of the events $v_{4i}$ and $w_i$.
This switch is necessary to satisfy Condition~\ref{con:manual}.\ref{con:six}.
However, similar to the TS $T_{i,0}$, the $T'_{i,0}$ has the sequence $P_i$.
The initial state of $T'_{i,0}$ is $t_{i,0,17}$ and, again, it uses $\zeta_i$'s variables as events.
Additionally, the union $U^{\sigma_2}_\varphi$ installs for clause $\zeta_i$ also the TSs $T_{i,1}, T_{i,2}, T_{i,3}$, originally introduced for $U^{\sigma_1}_\varphi$.
If $(sup, sig)$ is a region of $U^{\sigma_2}_\varphi$ that inhibits $k$ at $h_{0,2}$ then, by the discussions of the former section, we have either $sig(k)=\used$ or $sig(k)=\free$.
In the following we argue that if $sig(k)=\used$ then $M=\{X\in V(\varphi) \mid sig(X)=\set\}$ is a one-in-three model of $\varphi$ and, otherwise, $M=\{X\in V(\varphi) \mid sig(X)=\res \}$ is a one-in-three model of $\varphi$.

If $sig(k)=\used$ then $sup(t_{i,0,1})=sup(t_{i,0,15})=1$.
Thus, again by $V\subseteq sig^{-1}(\swap)$, $W\cap sig^{-1}(\swap)=\emptyset$ and $Acc\cap sig^{-1}(\swap)=\emptyset$ we obtain that the image of $P_i $ of $T'_{i,0}$ has a subsequence of $\tau$ that starts at $0$ and terminates at $1$.
Moreover, by $Acc\cap sig^{-1}(\swap)=\emptyset$, the signature of an event of $X_{i,0}, X_{i,1}, X_{i,2}$ has to realize this state change from $0$ to $1$ in $\tau$.
If $sig(X)\in \{\nop, \swap, \used, \free\}$ then $s\edge{X}s'\edge{X}s''$, where $s,s',s''$ are pairwise different, implies $sup(s)=sup(s'')$ and if $sig(X)=\res$ then $sup(s'')=0$.
Thus, there has to be an event $X\in \{X_{i,0}, X_{i,1}, X_{i,2}\}$ such that $sig(X)=\set$.
By the already discussed functionality of $T_{i,1}, T_{i,2}$ and $T_{i,3}$ we obtain that $sig(Y)\not=\set$ for $Y\in \{X_{i,0}, X_{i,1}, X_{i,2}\}\setminus \{X\}$.
By the arbitrariness of $i$ this is simultaneously true for all clauses, thus, $M=\{X\in V(\varphi) \mid sig(X)=\set\}$ is a one-in-three model of $\varphi$.

If $sig(k)=\free$ then, by similar arguments, we conclude that the image of the subsequence $t_{i,0,13}\edge{X_{i,2}}\dots\edge{X_{i,0}}t_{i,0,3}$ is a sequence of $\tau$ that starts at $1$ and terminates at $0$.
Moreover, there has to be an event $X\in \{X_{i,0}, X_{i,1}, X_{i,2}\}$ such that $sig(X)=\res$ which realizes the state change from $1$ to $0$ in $\tau$.
Using again the TSs $T_{i,1}, T_{i,2}, T_{i,3}$ and $V\subseteq sig^{-1}(\swap)$ we obtain that $sig(X)=\res$ for $X\in \{X_{i,0}, X_{i,1}, X_{i,2}\}$ implies that $sig(Y)\not=\res$ for $Y\in \{X_{i,0}, X_{i,1}, X_{i,2}\}\setminus \{X\}$.
Consequently, $M=\{X\in V(\varphi) \mid sig(X)=\res\}$ selects exactly one variable of every clause $\zeta_i$ of $\varphi$, which makes it a sought model.

%%%%%%%%%%%%%%%%%%%%%%%%%%%%%%%%%%%%%%%%%%%%%%%%%%%%%%
\subsection{Details for Condition~\ref{con:manual}.\ref{con:six}}%
%%%%%%%%%%%%%%%%%%%%%%%%%%%%%%%%%%%%%%%%%%%%%%%%%%%%%%

In this section, we prove that $U^{\sigma_1}_{\varphi}$ satisfies Condition~\ref{con:manual}.\ref{con:six}.
Due to space limitation, we can not present the proof for $U^{\sigma_2}_{\varphi}$ which, however, can be found in the technical report~\cite{DBLP:journals/corr/abs-1806-03703}.

In the remainder of this section let $\tau\in\sigma_1$.
To show the $\tau$-ESSP of $U^{\sigma_1}_{\varphi}$ we restrict ourselves for every $e\in E(U^{\sigma_1}_{\varphi})$ to the presentation of $\tau$-regions that inhibits $e$ in the gadgets $A$ of $U^{\sigma_1}_\varphi$ that actually apply $e$: $e\in E(A)$.
It is easy to see, that $e$ is inhibitable at the remaining states:
In accordance to Section~\ref{sec:compressed}, we use simply the region $(sup, sig)$ which is defined by $sup=S(U^{\sigma_1}_\varphi)\setminus \{s \in S(A)\mid e\in E(A), A\in U^{\sigma_1}_\varphi\}$.

In the following, we firstly provide a sequence of lemmata that altogether prove the $\tau$-ESSP of $U^{\sigma_1}_\varphi$.
Each lemma covers a whole set of events and presents regions that inhibits these events.
The regions are presented in accordance to Section~\ref{sec:compressed} and by "\textbf{targets}" we refer to the states at which the current investigated event is inhibited.
Secondly, we argue that the already presented regions are sufficient for the $\tau$-SSP, too.

\begin{lemma}
The event $k$ is inhibitable in $U^{\sigma_1}_\varphi$.
\end{lemma}
\begin{proof}
$R^k_1$ is defined by \textbf{sup}: $\{d_{j,3},d_{j,4}, d_{j,5} \mid j\in \{0,\dots, 3m-1\} \}$, $\{g_{j,2}, g_{j,3}, g_{j,4} \mid 0\leq j\leq m-1\}$, $\{t_{i,0,3}, t_{i,0,4}, t_{i,0,5}, t_{i,0,7}, t_{i,0,8}, t_{i,0,9},t_{i,0,11}, t_{i,0,12}, t_{i,0,13}\mid i\in \{0,\dots, m-1\}\}$, $\{f_{1,2}, f_{2,2},\dots, f_{2,6}\}$  

\noindent
\textbf{targets}: $\{t_{i,0,3}, t_{i,0,4}, t_{i,0,5}, t_{i,0,7}, t_{i,0,8}, t_{i,0,9},t_{i,0,11}, t_{i,0,12}, t_{i,0,13}\mid 0\leq i\leq m-1\}$, $\{d_{0,5},\dots, d_{3m-1,5}\}$, $\{f_{2,3}, f_{2,4}\}$, $\{f_{2,3}, f_{2,4}\}$

Let $i,j,\ell\in \{0,\dots, m-1\}$ be pairwise distinct and $\alpha,\beta\in \{0,1,2\}$ such that $X_{i,0}=X_{j,\alpha}=X_{\ell,\beta}$.
$R^k_2$ is defined by \textbf{sup}: $\{ t_{i,0,3}, t_{i,0,4}, t_{i,0,6}, \dots, t_{i,0,14}, t_{i,1,0},\dots, t_{i,1,4}, t_{i,2,0}, \dots, t_{i,2,4}\}$, $\{g_{i,3}, g_{i,4}, g_{i,5}\}$, $S(T_{\ell,1}),S(T_{\ell,2}), S(T_{\ell,3})$, $S(T_{j,1}),S(T_{j,2}), S(T_{j,3})$, $\{s, s', s''\mid s\edge{X_{j,\alpha}}s'\edge{X_{j,\alpha}}s'', s\edge{X_{\ell,\beta}}s'\edge{X_{\ell,\beta}}s''\}$, \\  $\{d_{x,2},\dots, d_{x,4}\mid x\in \{0,\dots, 3m-1\}\setminus \{3i, 3j+\alpha, 3\ell+\beta\}\}$, $\{g_{x,2},\dots, g_{x,4}\mid x\in \{0,\dots, m-1\}\setminus \{i, j, \ell\}\}$,  $\{d_{3j+\alpha,3},\dots,d_{3j+\alpha,5}, d_{3\ell+\beta,3},\dots,d_{3\ell+\beta,5}\}$

\noindent
\textbf{targets}: $t_{i,0,6}, t_{i,0,10}, t_{i,0,14}$ and $g_{i,5}$

Let $M$ be a one-in-three model of $\varphi$.
The event $k$ is inhibited at the other relevant states (of TSs that apply $k$) by the following region $R^k_3$ which is the only region that actually depends on the existence of $M$.
More exactly, $R^k_3$ is defined by the support \textbf{sup}: $S_H\cup S_F \cup S_G \cup S_D \cup S_0 \cup \dots \cup S_{m-1}$ where 
$S_H=\{h_{j,2} \mid j\in \{0,\dots, 4m-1\}\}$, $S_F=\{f_{0,2}, f_{0,5}, f_{0,8}, f_{1,2}, f_{1,5}, f_{2,2}, f_{2,5}, f_{2,6}, f_{2,9}\}$, $S_G=\{g_{j,2}, g_{j,3}, g_{j,4}, g_{j,8} \mid j\in \{0,\dots, m-1\}\}$, $S_D=\{ d_{j,2}, d_{j,3}, d_{j,4}, d_{j,8} \mid j\in \{0,\dots, 3m-1\} \}$ and for $i\in \{0,\dots,m-1\}$ the set $S_i$ is defined as follows:

\[S_i=
\begin{cases}
\{t_{i,0,2}, t_{i,0,3}, t_{i,0,4},t_{i,0,17}, t_{i,1,0}, t_{i,1,1}, t_{i,2,0}, t_{i,2,1}\}, & \text{ if } X_{i,0}\in M\\
\{t_{i,0,2}, \dots, t_{i,0,8}, t_{i,1,2}, t_{i,1,3}, t_{i,3,0}, t_{i,3,1}\}, & \text{ if } X_{i,1}\in M\\
\{t_{i,0,2}, \dots, t_{i,0,12}, t_{i,2,2}, t_{i,2,3}, t_{i,3,2}, t_{i,3,3}\}, & \text{ if } X_{i,2}\in M\\
\end{cases}
\]
\end{proof}

\begin{lemma}
The events $v_0,\dots, v_{4m-1}$ are inhibitable in $U^{\sigma_1}_\varphi$.
\end{lemma}
\begin{proof}

Let $j\in \{0,\dots, 4m-1\}$.
The following regions $R^v_1$ and $R^v_2$ inhibit $v_j$ at the presented targets.

Region $R^v_1$ is defined by \textbf{sup}: $\{h_{i,0}, h_{i,1} \mid i\in \{0,\dots, 4m-1\}\}$, $\{f_{0,2}, f_{0,3}, f_{0,4}\}$, 
\textbf{targets}: $h_{j,0}, h_{j,1}$\\

Region $R^v_2$ is defined by \textbf{sup}: $\{h_{i,0}, h_{i,1} \mid i\in \{0,\dots, 4m-1\}\}$, $\{f_{0,2}, f_{0,3}, f_{0,4}\}$\\
\noindent
\textbf{targets}: $\{h_{j,4}, h_{j,5}\}$, $\{t_{i,0,0}, t_{i,0,16}, t_{i,0,17} \mid i\in \{0,\dots, m-1\} \}$\\

Let $i,j,\ell\in \{0,\dots, m-1\}$ be pairwise distinct and $\alpha,\beta\in \{0,1,2\}$ such that $X_{i,0}=X_{j,\alpha}=X_{\ell,\beta}$.
The inhibition of $v_{4i}$ in $T_{i,0}$:
For $t_{i,0,0}, t_{i,0,16}, t_{i,0,17}$ use $R^v_2$, for $\{t_{i,0,3}, \dots, t_{i,0,13}\}\setminus \{t_{i,0,6}, t_{i,0,10}\}$ use $R^k_2$.
For $t_{i,0,6}, t_{i,0,10}, t_{i,0,14}, t_{i,0,15}$ use the region $R^v_3$ which is defined by 
\textbf{sup}:\\ $\{t_{i,0,3}, t_{i,0,4}, t_{i,0,6}, t_{i,0,10}, t_{i,0,14}, \dots, t_{i,0,17}\}$, $\{d_{j,6}, d_{j,7}, d_{j,8}\mid j\in \{4i, 4i+1, 4i+2\}\}$, $S(T_{i,1}), S(T_{i,2})$, $S(T_{j,0}),\dots, S(T_{j,3})$, $S(T_{\ell,0}),\dots, S(T_{\ell,3})$

For $i\in \{0,\dots, m-1\}$ and $\alpha\in \{1,2,3\}$ the region $R^v_4$ defined by \textbf{sup}: $\{t_{j,\alpha,0}, t_{j,\alpha,3}, t_{j,\alpha,4} \mid j\in \{0,\dots, m-1\}, \alpha\in \{1,2,3\} \}$, $\{t_{j,0,4}, t_{j,0,8}, t_{j,0,12}  \mid j\in \{0,\dots, m-1\} \}$ inhibits $v_{4i+\alpha}$ in $T_{i,\alpha}$ at the 
\textbf{targets}:  $\{t_{i,\alpha,0}, t_{i,\alpha,3}, t_{i,\alpha,4} \}$\\
\end{proof}

\begin{lemma}
The events $y_0,\dots, y_{m-1}$ and $p_0,\dots, p_{3m-1}$ are inhibitable in $U^{\sigma_1}_\varphi$.
\end{lemma}
\begin{proof}
Let $j\in \{0,\dots, m-1\}$.
The inhibition of $y_j$ in $G_j$ works as follows:
For $y_j$ at $g_{j,0}, g_{j,7}, g_{j,8}$ use the region $R^v_2$.
For $g_{j,6}$ use the region $R^y_1$ defined by \textbf{sup}: $\{t_{j,0,15}, t_{j,0,16}, t_{j,0,17}\}$, $\{g_{j,6}, g_{j,7}, g_{j,8}\}$.
For the remaining state $g_{j,3}$ use the region $R^y_2$ defined by \textbf{sup}: $\{f_{2,4}, f_{2,5}\}$, $\{d_{x,3}\mid x\in \{0,\dots, 3m-1\}\}$, $\{g_{x,3}\mid x\in \{0,\dots, m-1\}\}$.

Let $j\in \{0,\dots, 3m-1\}$.
The inhibition of $p_j$ in $D_j$ works as follows:
For the states $d_{j,0}, d_{j,7}, d_{j,8}$ use the region $R^v_2$ and for $d_{j,3}$ use $R^y_2$.
Finally, for $d_{j,6}$ use $R^p_1$ which is defined by \textbf{sup}$: \{d_{j,6}, d_{j,7}, d_{j,8}\}\cup S$, where $S=\{t_{i,0,3},t_{i,0,3}, T_{i,0,5}\}$ if $j=3i$ and $S=\{t_{i,0,7},t_{i,0,8}, t_{i,0,9}\}$ if $j=3i+1$ and $S=\{t_{i,0,11},t_{i,0,12}, t_{i,0,13}\}$ if $j=3i+2$. 
\end{proof}

\begin{lemma}
The event $z$ is inhibitable in $U^{\sigma_1}_\varphi$.
\end{lemma}
\begin{proof}
Region $R^z_1$ is defined by \textbf{sup}: $\{f_{j,0}, f_{j,1}, f_{j,2}\mid j\in \{0,1,2\}\}$, $\{d_{j,0}, d_{j,1}, d_{j,5},\dots, d_{j,8}\mid j\in \{0,\dots, 3m-1\}\}$, $\{g_{j,0}, g_{j,1}, g_{j,5},\dots, g_{j,8}\mid j\in \{0,\dots, m-1\}\}$, $\{f_{2,7}, f_{2,8}, f_{2,9}\}$ and inhibits $z$ at all relevant states.
\end{proof}

\begin{lemma}
The event $m$ is inhibitable in $U^{\sigma_1}_\varphi$.
\end{lemma}
\begin{proof}
Region $R^m_1$ is defined by \textbf{sup}: $\{f_{0,0} , f_{0,3}, f_{0,6}\}$, $\{f_{1,0}, f_{1,4}, f_{1,5}\}$, $\{f_{2,1}, f_{2,2}, f_{2,5}, f_{2,6}, f_{2,7}\}$, $\{h_{j,0}, h_{j,4}, h_{j,5} \mid j\in \{0,\dots, 4m-1\}\}$, $\{d_{j,0}, d_{j,7}, g_{\ell,0}, g_{\ell,7}\mid j\in \{0,\dots, 3m-1\}, \ell\in \{0,\dots, m-1\}\}$, $\{t_{i,0,0}, t_{i,0,16} \mid i\in \{0,\dots, m-1\}$, 
\textbf{targets}: $\{h_{j,0}, h_{j,4}, h_{j,5} \mid j\in \{0,\dots, 4m-1\}\}$, $\{f_{0,0} , f_{0,3}, f_{0,6}\}$, \\

Region $R^m_2$ is defined by \textbf{sup}: $sup(R^m_1)\setminus \{f_{0,6},f_{2,5}, f_{2,6}, f_{2,7}\}$, $\{f_{0,7}, f_{0,8}, f_{2,7}, f_{2,8}\}$, 
\textbf{targets}:  $\{f_{0,7} , f_{0,8}\}$

Region $R^m_3$ is defined by \textbf{sup}: $\{h_{j,3}, h_{j,4}, h_{j,5}\mid j\in \{0,\dots, 4m-1\}\}$, $\{t_{i,j,0}, t_{i,j,1}\mid j\in \{0,\dots,3\}, i\in \{0,\dots, m-1\}\}$, 
\textbf{targets}:  $\{h_{j,3}\mid j\in \{0,\dots, 4m-1\}\}$
\end{proof}

\begin{lemma}
The events $w_0,\dots, w_{m-1}$ are inhibitable in $U^{\sigma_1}_\varphi$.
\end{lemma}
\begin{proof}
Let $i\in \{0,\dots, m-1\}$.
Region $R^k_3$ inhibits $w_i$ at $g_{i,2}, t_{i,0,17}$ and $R^k_2$ at $t_{i,0,13}$ and $R^v_2$ at $g_{i,7}, g_{i,8}, t_{i,16}$.

Let $i,j,\ell\in \{0,\dots, m-1\}$ be pairwise distinct.
The following support defines region $R^w_1$ that inhibits $w_i$ at the (remaining) states $t_{i,0,0}, \dots, t_{i,0,12}$: 
\textbf{sup}: $\{t_{i,0,0}, \dots, t_{i,0,12}\}$, $\bigcup^{m-1}_{j=0}(S(T_{j,1})\cup S(T_{j,2}) \cup S(T_{j,3}))$, $\bigcup^{m-1}_{i\not= j=0}S(T_{j,0})$, $S(F_2)$, $\bigcup^{3m-1}_{ j=0}S(D_j)$, $\bigcup^{m-1}_{ i\not=j=0}S(G_j)$, $\{g_{j,0}, g_{j,1}, g_{j,3}, g_{j,4}\}$
\end{proof}

\begin{lemma}
The events $X_0,\dots, X_{m-1}$ are inhibitable in $U^{\sigma_1}_\varphi$.
\end{lemma}
\begin{proof}
We let $i\in \{0,\dots, m-1\}$ and explicitly show the inhibition of $X_{i,0}$ in $T_{i,0}, \dots, T_{i,3}$.
By symmetry, the inhibition of $X_{i,1}$ and $X_{i,2}$ works perfectly similar. 
By the arbitrariness of $i$ this proves the lemma.
The corresponding support is defined by \textbf{sup}: $\{t_{i,0,0}, t_{i,0,1}, t_{i,0,2},  t_{i,0,6}, \dots, t_{i,0,17}\}$,   $\{t_{i,1,2}, t_{i,1,3},t_{i,1,4},  t_{i,2,2}, t_{i,2,3},t_{i,2,4}\}$, $\{h_{4i+1,0}, h_{4i+1,1}, h_{4i+2,0}, h_{4i+2,1}\}$, $\{d_{3i,6}, d_{3i,7}, d_{3i,8}\}$.
\end{proof}

\begin{lemma}
The events $q_0,\dots, q_3$ are inhibitable in $U^{\sigma_1}_\varphi$.
\end{lemma}
\begin{proof}
($q_0$):
Region $R^{q_0}_0$ defined by \textbf{sup}: $\{h_{j,0}, h_{j,1} \mid j\in \{0,\dots, 4m-1\}\}$, $\{f_{0,0}, f_{0,1}, f_{0,5},\dots,f_{0,8}\}$, 
\textbf{targets}: $S(F_0)\setminus \{f_{0,4}\}$ 

Region $R^{q_0}_1$ defined by \textbf{sup}: $\{h_{j,0}, h_{j,4}, h_{j,5}  \mid j\in \{0,\dots, 4m-1\}\}$, $\{t_{j,0,0}, t_{j,0,16}, t_{j,0,17} \mid j\in \{0,\dots, m-1\}\}$, $\{d_{j,0}, d_{j,7}, d_{j,8} \mid j\in \{0,\dots, 4m-1\}\}$, $\{f_{1,4}, f_{1,5}\}$, $\{f_{2,0}, f_{2,8}, f_{2,9}\}$, $\{g_{j,0}, g_{j,7}, g_{j,8} \mid j\in \{0,\dots, m-1\}\}$, $\{f_{0,0}, f_{0,4}, f_{0,5},f_{0,6}\}$, 
\textbf{target}: $ \{f_{0,4}\}$ 

Region $R^{q_0}_2$ defined by \textbf{sup}: $\{f_{0,6}, f_{0,7}, f_{0,8}\}$, $\{f_{1,0}, f_{1,1}\}$, $\{f_{2,0}, f_{2,1}, f_{2,5},\dots, f_{2,9}\}$, 
\textbf{targets}: $ \{f_{2,0}, f_{2,1}, f_{2,5},\dots, f_{2,9}\}$, 

Region $R^{q_0}_3$ defined by \textbf{sup}: $\{f_{2,4}, f_{2,5}\}$, $\{d_{j,3} \mid j\in \{0,\dots, 3m-1\} \}$, $\{ g_{\ell,3} \mid  \ell\in \{0,\dots, m-1\}\}$, 
\textbf{target}: $ \{f_{2,4}\}$ 

\noindent
$(q_1)$:
Region $R^{q_1}_0$ defined by \textbf{sup}: $\{h_{j,0}, h_{j,1} \mid j\in \{0,\dots, 4m-1\}\}$, $\{f_{0,2}, f_{0,3}, f_{0,4}\}$, 
\textbf{targets}: $ \{f_{0,2}, f_{0,3}\}$ 

Region $R^{q_1}_1$ defined by \textbf{sup}: $\{f_{0,0}, f_{0,1}, f_{0,2}, f_{2,0}, f_{2,1}, f_{2,2}\}$, 
\textbf{targets}: $ \{f_{0,0}\}$ 

Region $R^{q_1}_3$ defined by \textbf{sup}: $\{t_{j,0,0}, t_{j,0,16} \mid j\in \{0,\dots, m-1\}\}$, $\{ h_{j,1}, h_{j,4} \mid j\in \{0,\dots, 4m-1\}\}$, $\{ f_{0,1}, f_{0,4},f_{0,7}, f_{0,8}\}$,  $\{f_{1,1}, f_{1,3}, f_{2,2}\}$, $\{d_{j,0}, d_{j,7} \mid j\in \{0,\dots, 3m-1\} \}$,  $\{g_{\ell,0},g_{\ell,7} \mid \ell\in \{0,\dots, m-1\}\}$, 
\textbf{targets}: $ \{f_{0,1}, f_{0,7}, f_{0,8}\}$ 

Region $R^{q_1}_4$ defined by \textbf{sup}: $ S(F_2)\setminus \{f_{2,3}, f_{2,5}\}$, $\{d_{j,3} \mid j\in \{0,\dots, 3m-1\} \}$, $\{ g_{\ell,3} \mid \ell\in \{0,\dots, m-1\}\}$, 
\textbf{targets}: $ S(F_2)$

$(q_2)$:
Region $R^{q_2}_0$ defined by \textbf{sup}: $\{f_{0,0}, f_{0,1}, f_{0,2}, f_{2,3}, \dots, f_{2,9}\}$, 
\textbf{targets}: $ \{f_{2,3}, \dots, f_{2,9} \}$

Region $R^{q_2}_1$ defined by \textbf{sup}: $\{t_{j,0,0}, t_{j,0,16} \mid j\in \{0,\dots, m-1\}\}$, $\{ h_{j,0}, h_{j,4} \mid j\in \{0,\dots, 4m-1\}\}$, $ \{f_{0,0}, f_{0,6}\}$, $\{f_{1,0}, f_{1,4}, f_{1,5}  \}$, $\{f_{2,0}, f_{2,3}, f_{2,5}, f_{2,8}, f_{2,9}\}$, $\{d_{j,3} \mid j\in \{0,\dots, 3m-1\} \}$, $\{ g_{\ell,3} \mid \ell\in \{0,\dots, m-1\}\}$,
\textbf{targets}: $ \{f_{1,0}, f_{1,4}, f_{1,5}, f_{2,0} \}$ 

Region $R^{q_2}_2$ defined by \textbf{sup}: $ \{f_{1,3}, f_{1,4}, f_{1,5}, f_{2,7}, f_{2,8}, f_{2,9} \}$, 
\textbf{targets}: $\{f_{1,3}\}$

$(q_3)$:
We reuse Region $R^{q_2}_1$ for the \textbf{targets}: $ \{f_{1,0}, f_{1,4}, f_{1,5}, f_{2,5}, f_{2,8} , f_{2,9}\}$. 

Region $R^{q_3}_0$ defined by \textbf{sup}: $ \{ f_{1,0}, f_{1,1}, f_{2,0}, f_{2,1} \}$, 
\textbf{targets}: $\{f_{1,0}, f_{1,1}\}$ 

Region $R^{q_3}_1$ defined by \textbf{sup}: $ \{ f_{0,6}, f_{0,7}, f_{0,8}\}$,  $\{f_{2,0}, \dots, f_{2,4}\}$, 
\textbf{targets}: $\{f_{2,0}, \dots, f_{2,4}\}$ 
\end{proof}

\begin{lemma}\label{lem:event_a}
The events $a_0,\dots, a_{3m-1}$ are inhibitable in $U^{\sigma_1}_\varphi$.
\end{lemma}
\begin{proof}

Let $i\in \{0,\dots, 3m-1\}$.
The following two regions $R^a_0$ and $R^a_1$ inhibit $a_i$ in $D_i$.

Region $R^a_0$ defined by \textbf{sup}: $\{d_{j,0}, d_{j,2}, d_{j,3}, d_{j,4}, d_{j,7}, d_{j,8} \mid j\in \{0,\dots, 3m-1\}\}$, $\{t_{j,0,0}, t_{j,0,16} \mid j\in \{0,\dots, m-1\}\}$, $\{ h_{j,0}, h_{j,4} \mid j\in \{0,\dots, 4m-1\}\}$, $\{f_{0,0}, f_{0,4}, f_{0,5}, f_{0,6}, f_{1,0}, f_{1,4}\}$, $\{g_{j,0}, d_{j,7} \mid j\in \{0,\dots, m-1\}\}$,
\textbf{targets}: $\{d_{j,0}, d_{j,2}, d_{j,3}, d_{j,4}, d_{j,7}, d_{j,8}\}$

Region $R^a_1$ defined by \textbf{sup}: $S(F_2)$, $\{d_{j,0}, d_{j,1}, d_{j,3}, d_{j,4}\mid  j\in \{0,\dots, 3m-1\} \}$, $\bigcup_{j=0}^{m-1}S(G_j)$, 
\textbf{targets}: $\{d_{j,1}\}$

The following three regions $R^a_2, R^a_3$ and $R^a_4$ are dedicated to the inhibition of $a_{3i}$ at the remaining states of $T_{i,0}$.
Notice that the region $R^v_4$ inhibits $a_{3i}$ at $t_{i,0,4}$.

Region $R^a_2$ defined by \textbf{sup}: all targets plus $\{h_{4i,0}, h_{4i,1}, h_{4i,2}\}$, $\{d_{j,6}, d_{j,7}, d_{j,8}\mid j \in \{3i+1, 3i+2\} \}$, $\{g_{i,6}, g_{i,7}, g_{i,8}\}$, 
\textbf{targets}: $\{t_{i,0,0}, t_{i,0,1}, t_{i,0,7}, t_{i,0,8}, t_{i,0,9} , t_{i,0,11}, t_{i,0,12}, t_{i,0,13}, t_{i,0,15}, t_{i,0,16}, t_{i,0,17} \}$

Region $R^a_3$ defined by \textbf{sup}: $\{  t_{i,0,7}, t_{i,0,8}, t_{i,0,10},  t_{i,0,14}, \dots , t_{i,0,17} \}$, $\bigcup^{m-1}_{j=0}(S(T_{j,1})\cup S(T_{j,2})\cup S(T_{j,3}) )$, $\bigcup^{m-1}_{i\not=j=0}S(T_{j,0})$, $\{d_{j,6}, d_{j,7}, d_{j,8}\mid j\in \{3i+1, 3i+2\} \}$, 
\textbf{targets}: $\{t_{i,0,10}, t_{i,0,14}\}$

The following three regions $R^a_5, R^a_6$ and $R^a_7$ are dedicated to the inhibition of $a_{3i+1}$ at the remaining states of $T_{i,0}$.
Notice that the region $R^v_4$ inhibits $a_{3i+1}$ at $t_{i,0,8}$.

Region $R^a_5$ defined by \textbf{sup}: all targets plus  $\{d_{j,6}, d_{j,7}, d_{j,8}\mid j\in \{3i, 3i+2\} \}$, $\{g_{i,6}, g_{i,7}, g_{i,8}\}$, 
\textbf{targets}: $\{t_{i,0,3}, t_{i,0,4}, t_{i,0,5}, t_{i,0,11}, t_{i,0,12} , t_{i,0,13}, t_{i,0,15}, t_{i,0,16}, t_{i,0,17}\}$,

Region $R^a_6$ defined by \textbf{sup}: $\{  t_{i,0,0},\dots,  t_{i,0,4} \}$, $\bigcup^{m-1}_{j=0}(S(T_{j,1})\cup S(T_{j,2})\cup S(T_{j,3}) )$, $\bigcup^{m-1}_{i\not=j=0} S(T_{j,0})$,
\textbf{targets}: $\{t_{i,0,0}, t_{i,0,1}, t_{i,0,2}\}$,

Region $R^a_7$ defined by \textbf{sup}: $\{  t_{i,0,11}, t_{i,0,12}, t_{i,0,14},\dots, t_{i,0,17} \}$, $\bigcup^{m-1}_{j=0}(S(T_{j,1})\cup S(T_{j,2})\cup S(T_{j,3}) )$, $\bigcup^{m-1}_{i\not=j=0} S(T_{j,0})$, $\{d_{j,6}, d_{j,7}, d_{j,8}\mid j\in \{3i+2\} \}$,
\textbf{targets}: $\{t_{i,0,14} \}$, 
Finally, the inhibition of $a_{3i+2}$ in $T_{i,0}$ works symmetrically to $a_{3i}$ and $a_{3i+1}$.
\end{proof}

\begin{lemma}[Without proof] 
The unique events, substituted by underscores, are inhibitable. 
\end{lemma}

Finally, we argue that the $\tau$-ESSP of $U^{\sigma_1}_\varphi$ implies its $\tau$-SSP:

\begin{lemma}
If $U^{\sigma_1}_\varphi$ has the $\tau$-ESSP then it has $\tau$-SSP. 
\end{lemma}
\begin{proof}
For a start, the initial state of any TS $A$ implemented by $U^{\sigma_1}_\varphi$ is obviously separable from all the other states of $A$.
Moreover, if $U^\sigma_{\varphi}$ has the $\tau$-ESSP then $k$ is inhibitable at $h_{0,2}$ and $\varphi$ has a one-in-three model $M$.
Thus, all the formerly presented regions exist.
We argue, that these regions justify the $\tau$-SSP of $U^{\sigma_1}_\varphi$.
To show we argue for every gadget $A$ of $U^{\sigma_1}_\varphi$ that it has the $\tau$-SSP by regions of $U^{\sigma_1}_\varphi$.
The following table present the corresponding regions.

\begin{center}
\begin{tabular}{  p{6cm}    p{6cm}}
TS & Separating Regions \\ \hline
%row 1
$H_j$, $j\in \{0,\dots, 4m-1\}$ 
&
$R^{q_0}_0, R^{q_0}_1, R^k_3$ \\ \hline
%row 2
$F_0$
&
$R^{q_0}_0, R^{q_0}_1, R^{q_0}_2, R^k_3$\\ \hline
%row 3
$F_1$
&
$R^{q_0}_0,R^{q_0}_2, R^k_3$ \\ \hline
%row 4
$F_2$
&
$R^{q_0}_0, R^{q_0}_2, R^{q_0}_3, R^{q_1}_3, R^{q_1}_4$\\ \hline
%row 5
$D_j$, $j\in \{0,\dots, 3m-1\}$ 
&
$R^{q_0}_1, R^{q_0}_3,  R^k_2,R^k_3, R^a_2$\\ \hline
%row 6
$G_j$,  $j\in \{0,\dots, m-1\}$ 
&
$R^{q_0}_1, R^{q_1}_0, R^k_2, R^k_3, R^a_2, R^w_1$\\ \hline
%row 7
$T^{\sigma}_{i,\alpha}$, $i\in \{0,\dots, m-1\}, \alpha\in \{1,2,3\}$ & 
$R^k_3,  R^v_4$\\ \hline
$T_{i,0}$, $i\in \{0,\dots, m-1\}$ & $R^k_3, R^v_4$ and the regions of Lemma~\ref{lem:event_a}
\end{tabular}
\end{center}
\end{proof}

%%%%%%%%%%%%%%%%%%%%
\section{Conclusion and Future Work}%
%%%%%%%%%%%%%%%%%%%%

In this paper, we continue our work of \cite{DBLP:conf/tamc/TredupR19,DBLP:conf/apn/TredupR19} and show the NP-completeness of $\tau$-feasibility for the boolean types $\tau=\{\nop,\swap\}\cup \omega$ with $\omega\subseteq \{\res,\set, \used,\free\}$, $\omega \cap \{\res,\set\}\not=\emptyset$ and $\omega \cap \{\used,\free\}\not=\emptyset$.
So far this settles the computational complexity of 120 of 256 possible boolean types of nets.
In addition, the presented reductions make sure that the resulting TSs are $2$-grade, which is a strong restriction.
This basically rules out the grade as a parameter for FPT approaches for all considered net types.
It remains future work to investigate the computational complexity of feasibility for the other 136 boolean types of nets.

%%%%%%%%%%%%%%%
\section{Acknowledgements}%
%%%%%%%%%%%%%%%
I'm grateful to the reviewers for their helpful comments.

%%%%%%%%%%%%%%%%
\bibliography{myBibliography}%
%%%%%%%%%%%%%%%%

%%%%%%%%%
\end{document}